\documentclass[twocolumn]{autart}

\usepackage{graphicx}

\usepackage{amsmath,amssymb,amsfonts}
\usepackage{array}
\usepackage{booktabs}
\usepackage{multirow}
\usepackage[noend]{algpseudocode}
\usepackage{algorithm}
\usepackage{subcaption}
\theoremstyle{remark}
\newtheorem{remark}{Remark}

\newtheorem{assumption}{Assumption}

\usepackage{comment}

\begin{document}

\begin{frontmatter}

\title{Parameter Privacy-Preserving Data Sharing: A Particle-Belief MDP Formulation
\thanksref{footnoteinfo}}

\thanks[footnoteinfo]{Corresponding author Kaidi Yang.}

\author[IORA]{Haokun Yu}\ead{yuhaokun@u.nus.edu},    
\author[CEE]{Jingyuan Zhou}\ead{jingyuanzhou@u.nus.edu},               
\author[CEE]{Kaidi Yang}\ead{kaidi.yang@nus.edu.sg}  

\address[IORA]{Institute of Operations Research and Analytics, National University of Singapore}  
\address[CEE]{Civil and Environmental Engineering, National University of Singapore}


\begin{abstract}                          
This paper investigates parameter-privacy-preserving data sharing in continuous-state dynamical systems, where a data owner designs a data-sharing policy to support downstream estimation and control while preventing adversarial inference of a sensitive parameter. This data-sharing problem is formulated as an optimization problem that trades off privacy leakage and the impact of data sharing on the data owner’s utility, subject to a data-usability constraint. 
We show that this problem admits an equivalent belief Markov decision process (MDP) formulation, which provides a simplified representation of the optimal policy. 
To efficiently characterize information-theoretic privacy leakage in continuous state and action spaces, we propose a particle-belief MDP formulation that tracks the parameter posterior via sequential Monte Carlo, yielding a tractable belief-state approximation that converges asymptotically as the number of particles increases. 
We further derive a tractable closed-form upper bound on particle-based MI via Gaussian mixture approximations, which enables efficient optimization of the particle-belief MDP. Experiments on a mixed-autonomy platoon show that the learned continuous policy substantially impedes inference attacks on human-driving behavior parameters while maintaining data usability and system performance.

\end{abstract}

\end{frontmatter}

\section{Introduction}
In networked control systems such as intelligent transport and smart buildings, there is a growing need for data sharing among agents. For instance, trajectory data released by mobility platforms like Uber has been extensively used by authorities to improve traffic efficiency and safety~\cite{Wang.Yang2024,wang2026fosteringdatacollaborationdigital}. However, sharing operational data can expose sensitive information about system states and parameters. In particular, trajectory data may reveal users’ travel patterns (e.g., destinations) or operators’ business secrets (e.g., algorithmic parameters). 

Existing privacy-preserving data sharing mechanisms mainly focus on protecting individual-level records (e.g., states) through cryptographic techniques~\cite{lu2018privacy,shang2022single,wang2022transmission,li2023dynamic,zou2025recursive} and differential privacy~\cite{chen2023differential,chen2023optimization,wang2023differential,wang2024differentially,zhao2014achieving}. Little attention has been devoted to protecting sensitive system parameters~\cite{nekouei2021randomized,farokhi2019ensuring,weng2025optimal,ziemann2020parameter}, such as control gains and operational constraints. These parameters are particularly critical for agents such as mobility platforms, since the disclosure of these parameters can undermine business competitiveness and, in some cases, even expose system vulnerabilities to adversarial exploitation. 
It is worth noting that individual-level techniques such as differential privacy are designed to conceal data records within a population, which do not naturally extend to parameter privacy, as parameters lack such a population in which they can hide.

Several studies have explored parameter privacy-preserving data-sharing through information-theoretic formulations, where data are deliberately perturbed to limit the information an adversary can infer about sensitive parameters, with leakage measured by quantities such as mutual information (MI) or entropy. Representative approaches include stochastic kernels that maintain statistical dependencies with private parameters~\cite{bassi2018lossy}, Gaussian-based filters designed to mask the dynamics of linear Gaussian systems~\cite{ziemann2020parameter}, randomized parameter-generation methods that offer improved computational efficiency~\cite{nekouei2022model}, and dynamic programming formulations that achieve optimal utility-privacy trade-offs~\cite{weng2025optimal}.

Despite notable progress, current mechanisms for parameter privacy remain constrained by two fundamental issues. First, most existing mechanisms (e.g., \cite{weng2025optimal,erdemir2020privacy}) have only been implemented in low-dimensional settings with discretized states, often relying on exhaustive enumeration of past observations. Such formulations are inadequate in practice because (i) state variables such as speed, traffic flow, and voltages often evolve in continuous domains, and (ii) discretization induces an exponential complexity growth when evaluating MI, rendering the methods computationally prohibitive in higher dimensions. Second, many studies assume that the data provider’s operation is unaffected by its data-sharing mechanism, effectively ignoring the feedback loop induced by the use of shared data~\cite{kawano2021modular}. In practice, however, released data are consumed by other agents whose decisions reshape the provider’s operating environment. For example, the released trajectory data may inform congestion price strategies of traffic authorities, which subsequently affect the operations of mobility platforms.

\emph{Statement of Contribution}. This paper addresses these limitations by developing a particle-belief framework for parameter privacy-preserving data sharing. The proposed framework aims to solve an optimization problem that trades off (i) parameter privacy, measured by the MI between the parameter and the shared data, and (ii) the impact of the shared data on system utility, captured through the feedback loop induced by the use of shared data. Our contributions are threefold. First, we derive a belief Markov decision process (MDP) that provides a simplified, belief-space formulation of the data-sharing optimization problem, in which optimal policies are history-independent. We theoretically establish its equivalence to the original optimization problem.
Second, to overcome the intractability of belief updates in continuous state spaces, we propose a particle-belief formulation that approximates the exact belief distribution using a set of particles. We establish the convergence of the Bellman value induced by the particle belief to the value function of the original belief MDP as the number of particles increases. Third, to address the computational burden associated with continuous data-sharing policies, we design an efficient approximation framework that provides a tractable upper bound on MI via Gaussian mixture model approximations.

A preliminary version of this paper was presented at the 7th Annual Learning for Dynamics \& Control Conference~\cite{yu2025interaction}. 
This journal version broadens the problem setting, provides rigorous analysis of the belief MDP, extends the framework to continuous data sharing via particle filtering and Gaussian mixture approximations, establishes theoretical guarantees for the particle-belief formulation, and presents expanded simulation studies.

\section{Problem Statement}

\label{sec:Problem Statement}

Consider a data provider modeled as a discrete-time Markov process with state $X_t  \in \mathsf{X} \subset \mathbb{R}^{N}$ and control action $U_t \in \mathbb{R}^{M}$, over the horizon $\mathcal{T} = \{0,1,\cdots, T\}$ with sampling interval $\Delta T$. The system evolution is governed by a transition kernel $ p(X_{t+1}\mid \Theta,X_t,W_t)$ that combines the system dynamics $p(X_{t+1} \mid X_t, U_t)$ and a control policy $p(U_t \mid \Theta, X_t, W_t)$, where the sensitive parameter $\Theta \in \mathcal S_{\Theta}\subset \mathbb{R}^{P}$ is treated as a time-invariant \emph{secret} (e.g., control gains), and $W_t$ is external input from the environment. To enable data sharing while protecting $\Theta$, the data provider releases a distorted output $Y_t$ instead of $X_t$, following a data-sharing policy $\pi_t \in \Pi$.  

For clarity, we impose the following assumptions on the external input and adversary’s knowledge.

\begin{assumption}[Markovian external input] \label{asm:external}
    The external input $W_t$ follows a Markov process with transition $p(W_{t+1}\mid W_t, Y_t)$, which may be unknown to the data provider. 
\end{assumption}
Assumption~\ref{asm:external} allows the transition to depend on the shared data $Y_t$. This is reasonable, since $Y_t$ may be exploited by other parties in the ecosystem (e.g., traffic authorities using mobility data to optimize signal timings).

\begin{assumption}[Adversary's knowledge] \label{asm:adversary}
    The adversary is assumed to know (i) all conditional probabilities represented by $p(\cdot \mid \cdot)$, (ii) the prior distribution $p(\Theta)$ of the sensitive parameter $\Theta$ but not its realization, and (iii) the data-sharing policy $\pi_t$ at any time $t$. 
\end{assumption}
We make two remarks. First, (i) and (iii) reflect a conservative worst-case stance, i.e., privacy guarantees should not depend on obscuring the conditional model \cite{shannon1949communication}. Relaxing this assumption only strengthens privacy. Second, prior knowledge of $\Theta$ is reasonable, as adversaries often know plausible ranges and commonly used values of $\Theta$.

We design a data-sharing policy $\pi_t$ at any time $t$ that satisfies the following criteria. First, the policy must protect the sensitive parameter $\Theta$ by preventing accurate Bayesian inference of its true value from the released data over time. Second, the policy should preserve control performance, ensuring that the distorted data $Y_t$ does not significantly degrade system operations. Note that $Y_t$ can influence system performance by negatively affecting the future states $X_{>t}$ via external input $W_t$ (see Assumption~\ref{asm:external}). Third, the shared data should retain data usability, which requires bounding the deviation between the distorted data $Y_t$ and the true state $X_t$.

To meet these criteria, the design of the data-sharing policy is formulated as  
\begin{subequations}\label{eq: original_optimization}
\begin{align}
\min_{\boldsymbol{\pi}} ~
& \rho \, I^{\boldsymbol{\pi}}\!\left(\Theta; Y_{1:T}, W_{1:T}\right) + \sum_{t=1}^T \mathbb{E}^{\boldsymbol{\pi}}\!\left[ r\!\left(\theta^*, X_{t:t+1}, Y_t\right) \right],
\label{eq:obj_org} \\[1ex]
\text{s.t.} ~
& \mathbb{E}^{\pi_t} \!\left[ d(X_t, Y_t) \right] \leq \hat{D},
\quad t = 1, 2, \dots, T,
\label{eq:con_org}
\end{align}
\end{subequations}
where $\boldsymbol{\pi} = \{\pi_t\}_{t=1}^T$. The objective in \eqref{eq:obj_org} integrates a privacy metric and system cost with a weighting coefficient~$\rho$. Privacy is quantified by the MI between $\Theta$ and the released data sequence and external input pair $(Y_{1:T},W_{1:T})$, capturing information leakage about $\Theta$. 
The system cost $r(\theta^*, X_{t:t+1}, Y_t)$ is a continuous function in all its arguments and captures the effect of publishing $Y_t$ when transitioning from state $X_t$ to $X_{t+1}$, given the true parameter $\theta^*$. Constraint \eqref{eq:con_org} enforces a prescribed upper bound $\hat{D}$ on the expected distortion $d(X_t, Y_t)$  under policy $\pi_t$, where $d(\cdot,\cdot)$ is also continuous.

We simplify the calculation of $I(\Theta;Y_{1:T}, W_{1:T})$ using the chain rule of MI (see Appendix~\ref{appendix: MI simplification}) as 
\begin{align}
    I^{\boldsymbol{\pi}}(\Theta;Y_{1:T},W_{1:T}) = \sum_{t=1}^{T}I^{\boldsymbol{\pi}}(\Theta; Y_t \mid Y_{1:t-1}, W_{1:t}). \label{eq:MI2}
\end{align}

\section{Belief MDP Formulation}
\label{sec: A Computationally Efficient Privacy-Preserving Policy}
 
In this section, we derive a belief MDP formulation equivalent to the data-sharing optimization problem~\eqref{eq: original_optimization}. 
In the most general form $\pi_t(Y_t \mid \Theta, X_{1:t}, Y_{1:t-1}, W_{1:t})$, the policy $\pi_t$ generates distorted output $Y_t$ as a function of the parameter $\Theta$, the state sequence $X_{1:t}$, the external input sequence $W_{1:t}$, and the history of shared data $Y_{1:t-1}$. While this representation is optimal in the sense that it conditions on all available information up to time $t$, solving~\eqref{eq: original_optimization} with this policy becomes computationally intractable as $t$ increases due to growing dimensionality. To address this issue, we introduce an equivalent simplified representation of the optimal policy. 

We first establish in Theorem~\ref{th:simplified policy Theorem} that the optimal solution to problem~\eqref{eq: original_optimization} can be restricted to policies of the form $\pi^s_t(Y_t \mid \Theta, X_t, W_{1:t}, Y_{1:t-1}) \in \boldsymbol{\Pi^s}$, which no longer require conditioning on the full state history $X_{1:t-1}$. See Appendix~\ref{appendix:simplified policy Theorem} for proof.

\begin{thm}[Optimality of Simplified Policies] \label{th:simplified policy Theorem}
There exists a sequence of simplified policies $\boldsymbol{\pi}^{s,*} = \{\pi^{s,*}_t\}_{t=1}^T \in \boldsymbol{\Pi}^s$ that achieves optimality for problem~\eqref{eq: original_optimization}. In particular, $\boldsymbol{\pi}^{s,*}$ satisfies constraint~\eqref{eq:con_org} and the corresponding objective value $L(\boldsymbol{\pi^{s,*}})$ coincides with that of the optimal policy sequence $\boldsymbol{\pi^*}$, i.e., $L(\boldsymbol{\pi^{s,*}})=L(\boldsymbol{\pi^*})$. 
\end{thm}
 
By Theorem~\ref{th:simplified policy Theorem}, the search for the optimal policy can be restricted to the simplified class $\boldsymbol{\Pi^s}$. We therefore reformulate \eqref{eq: original_optimization} as  
\begin{align}
    &\min_{\boldsymbol{\pi^s} \in \boldsymbol{\Pi^s}}\sum_{t=1}^T \Bigg(\rho I^{\boldsymbol{\pi^s}}(\Theta; Y_t\mid Y_{1:t-1},W_{1:t})\notag\\& + \mathbb{E}^{\boldsymbol{\pi^s}}[r(\theta^*,X_{t:t+1},Y_t)] + \lambda\Big(\mathbb{E}^{\boldsymbol{\pi^s}}[d(X_t,Y_t)]-\hat{D}\Big)\Bigg)\label{eq:simplified_policy_optimization}
\end{align}
where the Lagrangian multiplier $\lambda$ is  associated with constraint~\eqref{eq:con_org}. We solve \eqref{eq:simplified_policy_optimization} using dynamic programming, where given any history $h_{t} = (w_{1:t}, y_{1:t-1})$, the cost-to-go function $V_t(h_t)$ is defined as the minimum expected cumulative cost from time $t$ onward: 
\begin{align}\label{eq:cost_to_go}
    V_t^{*}(h_{t}) = &\min_{\{\pi_k^s\}_{k=t}^T}  \Bigg[ \sum_{k=t}^T \Bigg( \rho I^{\boldsymbol{\pi^s}}(\Theta; Y_k \mid Y_{1:k-1}, W_{1:k})\notag\\ &~+ \mathbb{E}^{\boldsymbol{\pi^s}}[r( X_{k:k+1}, Y_k, \theta^*)]  \notag\\&~+ \lambda\Big(\mathbb{E}^{\boldsymbol{\pi^s}}[d(X_k,Y_k)]-\hat{D}\Big)\Bigg)\Bigm| h_{t} \Bigg]
\end{align}
with terminal condition $V^{*}_{T+1}(\cdot) = 0$. Thus, we can write the Bellman optimality equation as 
\begin{align}
\label{eq:Bellman original}
    V^{*}_t(h_{t}) = \min_{\pi^s_t(Y_t\mid\Theta,X_t,h_t)} &[C_t(h_{t},\pi^s_t) +\mathbb{E}(V^{*}_{t+1}(h_{t+1}\mid h_{t}))],
\end{align}
where $C_t(h_{t}, \pi^s_t)$ denotes the stage cost under policy $\pi^s_t$. See Appendix~\ref{appendix:proof for bellman equiv} for its specific expression.

Nevertheless, the simplified policy $\boldsymbol{\pi}^s$ still relies on the history $h_{t}$, causing the computational complexity to grow rapidly as $T$ increases. To address this issue, we reformulate the problem using a belief-state representation. Specifically, we map the history $h_t$ to a belief distribution $\beta_{t}(\Theta, X_t) = p(\Theta, X_t \mid h_{t})$, which compactly summarizes all past information as a probability distribution over the current states and the sensitive parameter. The belief state updates following Bayes' rule upon observing new information (i.e., the distorted data $y_t$ and external input $W_t$), written as
\begin{align}
\label{eq:belief_update}
     \beta_{t+1}(\Theta, X_{t+1}) \propto &\int_{x_t}p(X_{t+1}\mid \Theta,x_t,W_t)a(y_t\mid\Theta,x_t)\notag \\ 
     & ~\cdot\beta_{t}(\Theta, x_{t})dx_t
\end{align}
where $a(y_t\mid \Theta,x_t)= \pi^s_t(Y_{t}\mid \Theta,x_t,W_{1:t},Y_{1:t-1})$ denotes the policy. See Appendix~\ref{appendix:update for belief} for detailed derivation. 

From~\eqref{eq:belief_update}, it follows that $\beta_t$ depends only on the current observation $y_t$, the policy $\pi^s_t$, and the previous belief state $\beta_{t-1}$. Based on this, Lemma~\ref{le:bellman equiv} shows that the optimal cost-to-go function expressed in terms of the belief state $\beta_t$ is equivalent to that obtained under the full history $h_t$ (see Appendix~\ref{appendix:proof for bellman equiv} for proof).  
Note that this conclusion build upon the idea in~\cite{weng2025optimal} and extends it to explicitly incorporate the effect of data-sharing on expected control performance  $\mathbb{E}^{\boldsymbol{\pi^s}}[r(\theta^*,X_{t:t+1},Y_t)]$.

\begin{lem}[Bellman Equivalence on Belief States]\label{le:bellman equiv}
At any time $t$, given the history $h_t$, the optimal cost-to-go function $V_t^*(h_t)$ can be expressed solely in terms of the belief state $\beta_t$ ,  i.e., $V_t^*(h_t) = V_t^*(\beta_t)$
\end{lem}

It follows from Lemma~\ref{le:bellman equiv} that given $h_t$, the data-sharing policy $\pi^s_t(Y_t \mid \Theta, X_t, h_t)$ induces a Markov kernel $\mathcal{K}_t$ that maps the state pair $(\Theta, X_t)$ to a probability distribution over $Y_t$. For arbitrary measurable spaces of $(\Theta, X_t)$ and $Y_t$, this kernel $\mathcal{K}_t$ specifies a conditional distribution of $Y_t$ under policy $\pi^s_t$. Based on this representation, we can replace $h_t$ with the belief state $\beta_t$ and the policy $\pi_t^s$ with its induced Markov kernel $\mathcal{K}_t$ in \eqref{eq:Bellman original}, thereby converting the Bellman optimality equation to 
\begin{align}\label{eq: Bellman kernel}
    V_t^*(\beta_t) = \min_{\mathcal{K}_t} \big[ C_t(\beta_t, \mathcal{K}_t) + \mathbb{E}[V_{t+1}^*(\beta_{t+1}) \mid h_t] \big]. 
\end{align}
By \eqref{eq: Bellman kernel}, it follows that the optimal Markov kernel $\mathcal{K}^*_t = \arg \min V_t^*(\beta_t)$ is a measurable mapping $\beta\mapsto \mathcal K^*_t(\cdot\mid \cdot,\cdot,\beta)$; given the realized belief $\beta_t$, we use $\mathcal K^*_t(\cdot\mid \cdot,\cdot,\beta_t)$. Thus, we show in Theorem~\ref{th:optimal policy} that the data-sharing policy is uniquely specified by $\beta_t$, $\Theta$, and $X_t$.

\begin{thm}[Optimal Policy Equivalence]\label{th:optimal policy}
For any history $h_t$ with belief $\beta_t$ on $\mathcal S_\Theta\times\mathsf X$. Then there exists an optimal Markov kernel
$\mathcal K_t^*(\cdot\mid \theta,x,\beta)$ such that the optimal data-sharing policy satisfies
\begin{align}\label{eq:pi-opt-kernel}
\pi_t^{s,*}(Y_t\mid \Theta,X_t,h_t)
=
\mathcal K_t^*(Y_t \mid \Theta,X_t,\beta_t),
\end{align}
for all $(\theta,x)\in \mathcal S_{\Theta}\times\mathsf X$. 
Moreover, for each fixed belief $\beta$, the kernel $\mathcal K_t^*(\cdot\mid\cdot,\cdot,\beta)$
attains the minimum in the belief-state Bellman equation \eqref{eq: Bellman kernel}.
\end{thm}

Theorem~\ref{th:optimal policy} shows that the optimal data-sharing policy is completely determined by the current state $(\Theta, X_t)$ and the belief state $\beta_t$. Therefore, to avoid confusion, we write $\mathcal K_t^*(Y_t \mid \Theta,X_t,\beta_t)$ in place of $\pi^{s,*}_t(Y_t\mid \Theta, X_t,h_t)$ throughout the remainder of the paper.

In existing research, state $X_t$ and policy $\mathcal K_t^*(Y_t \mid \Theta,X_t,\beta_t)$ are modeled in discrete spaces, restricting the state representation and the shared data to predefined discrete sets. Extending the framework to continuous state and action spaces provides a more accurate representation of system dynamics and enables smoother, more adaptive data-sharing strategies.

\section{Particle-Belief Formulation and Upper Bound of MI in Continuous Spaces}
\label{sec: Particle-Belief Formulation}

Solving the optimal data-sharing policy can be challenging, especially in settings with continuous state and action spaces. To handle the continuous state space, we employ a PF to track the evolution of the belief state $\beta_t$ and accordingly reformulate the Bellman optimality equation~\eqref{eq: Bellman kernel}. To handle the continuous action space, we adopt a Gaussian data-sharing policy to generate the shared data. Since the exact computation of MI in continuous spaces is generally intractable, we develop an efficient method to estimate an upper bound on the MI by combining particle-based beliefs with the Gaussian policy. We prove that the Bellman value error induced by particle beliefs uniformly converges in probability to zero. We further derive a tractable MI estimator whose gap is bounded by cluster similarity/separation parameters (KL-bound $\kappa$, Bhattacharyya $\gamma$).

\subsection{Particle formulation of belief state}\label{subsec:PF}

The belief state $\beta_t$, defined in~\eqref{eq:belief_update} as the posterior distribution of $(\Theta, X_t)$, is generally intractable to compute in continuous spaces. To address this challenge, we employ a PF to approximate $\beta_t$  by an empirical probability measure $\hat{\beta}_t^{(N)}$ supported on a finite set of weighted particles $\{(\theta_i,x_{i,t})\}_{i=1}^N$, with associated normalized nonnegative weights $\{\omega_{i,t}\}_{i=1}^N$ satisfying $\sum_{i=1}^N\omega_{i,t}=1$: 
\begin{align}
\hat{\beta}_t^{(N)}(d\theta,dx) =\sum_{i=1}^N \omega_{i,t}\delta_{(\theta_i,x_{i,t})}(d\theta,dx),
\label{eq:particle_measure_def}
\end{align}
where $\delta_{(\theta,x)}$ denotes the Dirac measure (unit point mass) at $(\theta,x)$. Note that all expectations under $\hat\beta_t^{(N)}$ reduce to weighted sums, i.e., for any bounded measurable test function $\varphi$, 
\begin{align}
\int \varphi(\theta,x)\,\hat\beta_t^{(N)}(d\theta,dx)
~=~\sum_{i=1}^N \omega_{i,t}\,\varphi(\theta_i,x_{i,t}).
\label{eq:particle_integral_sum}
\end{align}
Unlike grid-based discretization on a fixed mesh, a PF uses adaptive weighted samples in continuous space, focusing on high-posterior regions to avoid dimensionality explosion and discretization bias. Particles propagate according to the system dynamics. Their weights are updated based on consistency with the latest observations, as formalized in Lemma~\ref{lem: weight_update} (see Appendix~\ref{appendix:weight_update_derivation} for proof). 
\begin{lem}[Recursive Weight Updates]\label{lem: weight_update} Given policy $\mathcal{K}_t(y_t\mid\theta_i, x_{i,t},\beta_t)$ and the shared data $y_t$, the weight $\omega_{i,t}$ of particle $i$ is updated recursively by $
    \omega_{i,t}\propto \omega_{i,t-1}\mathcal{K}(y_t\mid\theta_i, x_{i,t}, \beta_t)$, 
with $\sum_{i=1}^N \omega_{i,t} = 1 $.
\end{lem}
Following the particle weight update, a conditionally unbiased resampling method (e.g., sequential importance resampling in~\cite{rubin1981bayesian}) is applied to mitigate degeneracy in which only a few particles dominate, and diversity is lost. Resampling is triggered when a weight-degeneracy score $D(\tilde{\omega}_{t})$ drops below a threshold (e.g., the effective sample size $\hat{N}_{\text{eff}} = \frac{1}{\sum_{i=1}^N (\tilde{\omega}_{i,t})^2}$). After weight updates, particles are then propagated according to system dynamics.

Lemma~\ref{lem:pf-consistency} summarizes the weak consistency of the particle belief. See \cite[Theorem~1]{crisan2002survey} (also \cite{chopin2004central,moral2004feynman}) for proof. Here, $\Rightarrow$ denotes weak convergence of probability measures, i.e., for any bounded continuous function $\varphi$, $\int \varphi(\theta,x)\,\hat\beta_t^{(N)}(d\theta,dx) \rightarrow\int \varphi(\theta,x)\,\beta_t(d\theta,dx)$. For notational convenience, we use $\beta_t$ to denote both the belief state and the corresponding probability measure.  
\begin{lem}[Particle-belief consistency]\label{lem:pf-consistency}
Fix $t$. Under standard regularity conditions for Sequential Monte Carlo (SMC) methods (i.e.,  Feller state transition and a strictly positive, bounded, and continuous data-sharing policy), the particle belief
$\hat\beta_t^{(N)}$ is weakly consistent, i.e.,
\begin{align}\label{eq:belief-weak-conv}
\hat\beta_t^{(N)} \Rightarrow \beta_t \qquad \text{a.s.}
\end{align} 
\end{lem}

By replacing the true belief $\beta_t$ with the particle belief $\hat\beta_t^{(N)}$ and
using the particle belief-update mapping, we obtain the following particle (surrogate)
Bellman optimality equation:
\begin{align}\label{eq:particle_bellman}
V_t^{*,(N)}(\hat\beta_t^{(N)})
= &\min_{\mathcal K_t\in\mathcal A_t}
\Big\{ C_t^{(N)}(\hat\beta_t^{(N)},\mathcal K_t)
\notag\\&\qquad + \mathbb E\!\big[ V_{t+1}^{*,(N)}(\hat\beta_{t+1}^{(N)}) \mid h_t,\mathcal K_t \big] \Big\}.
\end{align}
It is crucial to ensure that substituting $\hat{\beta}_t$ for the true belief $\beta_t$ in the Bellman equation does not distort the optimality principle. Therefore, we next formally establish the asymptotic convergence of the particle-belief-based Bellman optimality equation to its true form as the number of particles tends to infinity.

\subsection{Asymptotic convergence of Bellman equation}
We now analyze the asymptotic convergence of the particle-belief-based Bellman optimality equation \eqref{eq:particle_bellman}. Without loss of generality, we assume $\Theta$ to take values in a finite discrete set $\mathcal S_{\Theta} = \{\theta^1,\dots,\theta^K\}$. 
For any $\theta$, let $\mathcal{I}_\theta=\{i\in\{1,\cdots,N\} \mid \theta_i = \theta\}$ be the set of particles with $\theta_i=\theta$. 
According to \eqref{eq:particle_integral_sum}, the marginal mass is given by
$\hat q_t^{(N)}(\theta)=\sum_{i\in \mathcal I_\theta}\omega_{i,t}$ and the conditional
empirical measure is 
$\hat b_t^{(N)}(dx\mid\theta)=\sum_{i \in \mathcal I_\theta}\frac{\omega_{i,t}}{\hat q_t^{(N)}(\theta)}\delta_{x_{i,t}}(dx)$. Let $q_0(\theta^j)>0$ for all $j=1,\dots,K$ denote the prior over each value satisfying $\sum_{j=1}^K q_0(\theta^j)=1$. 

At each time step $t$, we next show that as the number of particles $N$ tends to infinity, the optimal value function $V_t^{*,(N)}(\cdot)$ in the particle formulation \eqref{eq:particle_bellman} converges to the optimal value function $V_t^*(\cdot)$ in the original formulation \eqref{eq: Bellman kernel}. 
Throughout this subsection, we fix $t\in\{1,\dots,T\}$ and drop the subscript $t$ from the arguments of the kernel $\mathcal K_t$  to simplify notation whenever no ambiguity arises, while retaining the time index on $\mathcal K_t$. We further impose the following regularity assumptions. 

\begin{assum}[Regularity conditions]\label{assum:regularity} The following statements hold for the state space and data-sharing policy: 
\begin{enumerate}
    \item[\emph{(i)}]  The state space $\mathsf X$ is compact.

    \item[\emph{(ii)}] 
    
    For each $t$, the admissible data-sharing policies $\mathcal K_t(dy\mid \theta,x,\beta)$ admit strictly positive probability density functions $k_t(y\mid\theta,x,\beta;u)$, parameterized by a decision variable
    $u\in\mathcal U_t$, where $\mathcal U_t\subset\mathbb R^{d_u}$ is compact, i.e., $
    \mathcal K_t(dy\mid \theta,x,\beta)=k_t(y\mid\theta,x,\beta;u)dy$. 
    Moreover, for any fixed $(\theta,y,\beta,u)$, the map $x\mapsto k_t(y\mid\theta,x,\beta;u)$ is continuous on $\mathsf X$.
    \item[\emph{(iii)}]There exist measurable functions $\underline k_t,\overline k_t:\mathbb R^N\to(0,\infty)$ such that $
0<\underline k_t(y)\le k_t(y\mid\theta,x,\beta;u)\le \overline k_t(y),~\forall(\theta,x,\beta,u,y)$. 
These functions satisfy the integrability conditions $\displaystyle\int\overline k_t(y)\,dy <\infty$, $\displaystyle\int\overline k_t(y)\Big(1+\big|\log \underline k_t(y)\big| +\big|\log \overline k_t(y)\big|\Big)\,dy <\infty$, and $\displaystyle\int(1+\|y\|_2)\overline k_t(y)\,dy<\infty$.
\item [\emph{(iv)}]
There exists $L_k<\infty$ such that for all $u,v\in\mathcal U_t$ and all $(\theta,x,\beta,y)$
\[
|k_t(y\mid\theta,x,\beta;u)-k_t(y\mid\theta,x,\beta;v)|
\le L_k\|u-v\|\,\overline k_t(y).
\] 
\item [\emph{(v)}]
$r(\theta^*,x_{t:t+1},y_t)$ is continuous and bounded.
\item [\emph{(vi)}]
The transition $p(X_{t+1}\mid \Theta, X_t,W_t)$ is Feller, i.e., for any bounded continuous  function $g$, the mapping $x_t\mapsto \int g(x_{t+1})p(dx_{t+1}\mid \Theta, x_t,w_t)$ is continuous.
\end{enumerate}
\end{assum}

To establish the asymptotic convergence of the Bellman optimality equation under Assumption~\ref{assum:regularity}, we first show the asymptotic convergence of the stepwise cost in Lemma~\ref{lem:MI-conv} (see Appendix~\ref{appendix:lem_MI-conv} for proof). 
\begin{lem}\textbf{\emph{(Asymptotic convergence of stepwise cost)}}\label{lem:MI-conv}
Under Assumption~\ref{assum:regularity}, $\lim_{N\rightarrow\infty}C_t^{(N)}\big(\hat\beta_t^{(N)},\mathcal K_t\big) = C_t\big(\beta_t,\mathcal K_t\big)$ almost surely. 
\end{lem}

Then, the asymptotic convergence of the Bellman optimality equation is stated below. 
\begin{thm} \textbf{\emph{(Asymptotic convergence of Bellman equation)}}\label{thm: asymptotic convergence} 
\label{thm:asymptotic_convergence}
Under Assumption~\ref{assum:regularity}, $\lim_{N\to\infty}V_t^{*,(N)}\!\big(\hat\beta_t^{(N)}\big) = V_t^*\!\big(\beta_t\big)$ almost surely. 
\end{thm}

\textbf{Proof sketch for Theorem~\ref{thm: asymptotic convergence}.}
We use backward induction.
At the terminal time, the convergence holds trivially since $V_{T+1}^{*,(N)}\equiv V_{T+1}^*\equiv 0$.
At the induction step, assume $V_{t+1}^{*,(N)}(\hat\beta_{t+1}^{(N)})\to V_{t+1}^*(\beta_{t+1})$ almost surely. 
Define $Q_t^{(N)}(\hat\beta_t^{(N)},\mathcal K_t) := C^{(N)}_t(\hat \beta^{(N)}_t, \mathcal{K}_t) + \mathbb{E}[V_{t+1}^{*,(N)}(\hat \beta^{(N)}_{t+1}) \mid h_t,\mathcal K_t]$ and $ Q_t(\beta_t,\mathcal K_t) := C_t( \beta_t, \mathcal{K}_t) + \mathbb{E}[V_{t+1}^*(\beta_{t+1}) \mid h_t,\mathcal K_t]$. Recall $h_t=(w_{1:t}, y_{1:t-1})$ is the history. Then, 
\begin{align}
    V_t^{*,(N)}(\hat\beta_t^{(N)})=&\min_{\mathcal K_t\in\mathcal A_t}Q_t^{(N)}(\hat\beta_t^{(N)},\mathcal K_t)\notag \\
V_t^*(\beta_t)=&\min_{\mathcal K_t\in\mathcal A_t}Q_t(\beta_t,\mathcal K_t),\notag
\end{align}
By Assumption~(ii), we can replace
$\min_{\mathcal K_t\in\mathcal A_t}(\cdot)$ by $\min_{u_t\in\mathcal U_t}(\cdot)$ and replace $\mathcal K_t$ by $u_t$ in both $Q_t$ and $Q_t^{(N)}$. 
Then, by the triangle inequality,
\begin{align}
    &| V_t^{*,(N)}(\hat\beta_t^{(N)})-V_t^*(\beta_t)|\notag \\&\qquad\leq \sup_{ u_t\in \mathcal U_t}|Q_t^{(N)}(\hat\beta_t^{(N)}, u_t)- Q_t(\hat\beta_t, u_t)|.\label{eq:delta_N}
\end{align}
Define $\Delta Q_N := \sup_{u_t\in\mathcal U_t}|Q_t^{(N)}(\hat\beta_t^{(N)},u_t)-Q_t(\beta_t,u_t)|$. It suffices to show $\Delta_N\rightarrow 0$ almost surely. 
By the triangular inequality, we decompose $\Delta_N\le \Delta Q_N^{(C)}+\Delta Q_N^{(V)}$, where $\displaystyle \Delta Q_N^{(C)} =\sup_{u_t\in\mathcal U_t}|C^{(N)}_t(\hat\beta_t^{(N)},u_t)-C_t(\beta_t,u_t)|$ is the current cost gap, and  $\displaystyle  \Delta Q_N^{(V)} = \sup_{u_t\in\mathcal U_t}|
\mathbb E[V_{t+1}^{*,(N)}(\hat\beta_{t+1}^{(N)})\mid h_t,u_t]
-\mathbb E[V_{t+1}^{*}(\beta_{t+1})\mid h_t,u_t]
|$ is the future value gap.
We then establish the point convergence and uniform convergence of these gaps (see Appendix~\ref{app:proof_thm_asymp} for details).

\subsection{Upper bound on particle MI under Gaussian policy}
Prior work \cite{weng2025optimal,erdemir2020privacy} often discretizes shared data, leading to coarse, inaccurate representations and a combinatorial growth of possible data-sharing outcomes in high-dimensional spaces. 
To address this, we enable the sharing of continuous data by devising a Gaussian policy conditioned on the current state, private parameter, and belief state, i.e., $Y_t \sim\mathcal{N}(\mu_t(X_t,\Theta,\hat{\beta}_t^{(N)};u),\Sigma_t(X_t,\Theta,\hat{\beta}_t^{(N)};u))$, where the mean $\mu_t$ and covariance matrix $\Sigma_t$ are the output of a deterministic mapping (e.g., a neural network) with $u\in\mathcal{U}$ being parameters and $X_t,\Theta,\hat{\beta}_t^{(N)}$ being input. 

We next quantify the particle MI under such a policy. The per-step particle MI is expressed as $I^{(N)}_u(\Theta;Y_t\mid h_t)  = H^{(N)}(Y_t \mid h_t;u) - H^{(N)}(Y_t \mid \Theta, h_t;u)$, where $H^{(N)}$ denotes corresponding conditional differential entropy.

We calculate the induced marginal density of $Y_t$ given $h_t$ as the following Gaussian mixture: 
\begin{align}
p^{(N)}(y\mid h_t;u)
&=\int k_t(y\mid \theta,x,\hat\beta_t^{(N)};u)\,\hat\beta_t^{(N)}(d\theta,dx)\notag\\
&=\sum_{i=1}^N \omega_{i,t}\mathcal{N}(y;\mu_{i,t}(u), \Sigma_{i,t}(u)),
\label{eq:global_gmm_sum}
\end{align}
where $\mu_{i,t}(u) := \mu_t\!\big(\theta_i,x_{i,t},\hat\beta_t^{(N)};u\big)$ and $\Sigma_{i,t}(u) := \Sigma_t\!\big(\theta_i,x_{i,t},\hat\beta_t^{(N)};u\big)$ denote the mean and covariance obtained by evaluating the policy at particle~$i$. 

Recall the marginal mass
$\hat q_t^{(N)}(\theta)=\sum_{i \in \mathcal I_\theta}\omega_{i,t}$ and the conditional
empirical measure
$\hat b_t^{(N)}(dx\mid\theta)=\sum_{i \in \mathcal I_\theta}\frac{\omega_{i,t}}{\hat q_t^{(N)}(\theta)}\delta_{x_{i,t}}(dx)$. Given $\theta$, the conditional density of $Y_t$ is therefore a Gaussian mixture:
\begin{align}
p^{(N)}(y\mid \theta,h_t;u)
&=\int k_t(y\mid \theta,x,\hat\beta_t^{(N)};u)\,\hat b_t^{(N)}(dx\mid \theta)\notag\\
&=\sum_{i\in\mathcal I_\theta}\frac{\omega_{i,t}}{\hat q_t^{(N)}(\theta)}\mathcal N\big(y;\mu_{i,t}(u),\Sigma_{i,t}(u)\big).
\label{eq:theta_gmm}
\end{align}
Consequently, both $H^{(N)}(Y_t\mid h_t)$ and $H^{(N)}(Y_t\mid \Theta=\theta,h_t)$
are the differential entropies of Gaussian mixture models (GMMs).
Since the differential entropy of a general GMM does not admit a closed-form expression, 
we derive a computable
upper bound on the MI $I^{(N)}(\Theta;Y_t\mid h_t)$ to enable tractable planning and learning. 

\begin{thm}[KL/Chernoff Mixture Bound]\label{th: upp-bound} Denote $\mathcal N_{i,t}(\cdot;u):=\mathcal N\!\big(\cdot;\mu_{i,t}(u),\Sigma_{i,t}(u)\big)$ 
and $\bar\omega_{i,t}(\theta):=\omega_{i,t}/\hat q_t^{(N)}(\theta)$. 
Then, $\widehat I^{(N)}_u(\Theta;Y_t\mid h_t) := \widehat H^{(N)}_{\mathrm{KL}}(Y_t\mid h_t;u)
-\sum_{\theta\in \mathcal S_\Theta}\hat q_t^{(N)}(\theta)\cdot 
\widehat H^{(N)}_{C_\alpha}\!\big(Y_t\mid \theta,h_t;u\big)$ is an upper bound on the per-step particle MI $I^{(N)}_u(\Theta;Y_t\mid h_t)$, where 
\begin{align}
&\widehat H^{(N)}_{\mathrm{KL}}(Y_t\mid h_t;u)
=
\sum_{i=1}^N \omega_{i,t}H(\mathcal N_{i,t}(\cdot;u))
\notag\\&-\sum_{i=1}^N \omega_{i,t}\log\Bigg(\sum_{j=1}^N \omega_{j,t}
e^{-\mathrm{KL}\big(\mathcal N_{i,t}(\cdot;u)\|\mathcal N_{j,t}(\cdot;u)\big)}\Bigg),
\notag\\
&\widehat H^{(N)}_{C_\alpha}\!\big(Y_t\mid \theta,h_t;u\big)
=
\sum_{i\in\mathcal I_\theta}\bar\omega_{i,t}(\theta)H(\mathcal N_{i,t}(\cdot;u))
\notag\\&-\sum_{i\in\mathcal I_\theta}\bar\omega_{i,t}(\theta)\log\Bigg(\sum_{j\in\mathcal I_\theta}\bar\omega_{i,t}(\theta)
e^{-C_{\alpha}\big(\mathcal N_{i,t}(\cdot;u)\|\mathcal N_{j,t}(\cdot;u)\big)}\Bigg).\notag
\end{align}
\end{thm} 
The proof of Theorem~\ref{th: upp-bound} combines 
(i) a KL-based upper bound $H^{(N)}(Y_t\mid h_t;u) \leq \widehat H^{(N)}_{\mathrm{KL}}(Y_t\mid h_t;u)$, with $\mathrm{KL}(\cdot\|\cdot)$ indicating the KL divergence, 
and (ii) a Chernoff $\alpha$-divergence-based lower bound $H^{(N)}\!\big(Y_t\mid \Theta,h_t;u\big) \geq \sum_{\theta\in\mathcal S_\Theta}\hat q_t^{(N)}(\theta)\cdot 
\widehat H^{(N)}_{C_\alpha}\!\big(Y_t\mid \theta,h_t;u\big)$, with $C_{\alpha}(\cdot\|\cdot)$ indicating the Chernoff $\alpha$-divergence \cite{kolchinsky2017estimating}. Since all components are Gaussian,
both $\mathrm{KL}$ and $C_{\alpha}$ admit closed-form expressions, and hence the bound is fully computable.

We next analyze the gap between the particle MI and the established upper bound, denoted by $
\Delta I^{(N)}_u=\widehat I^{(N)}_u(\Theta;Y_t\mid h_t)-I^{(N)}_u(\Theta;Y_t\mid h_t)$. 
Following the tightness analysis in~\cite{kolchinsky2017estimating}, we further establish an upper bound for $\Delta I^{(N)}_u$. Let us partition particles into multiple clusters using a clustering function $g(\cdot)$ according to two criteria: (1) $\mathrm{KL}(\mathcal N_{i,t}(\cdot; u)\|\mathcal N_{j,t}(\cdot; u))\leq \kappa$ whenever $g(i)=g(j)$ for some small $\kappa$ and (2) $\mathrm{BD}(\mathcal N_{i,t}(\cdot; u)\|\mathcal N_{j,t}(\cdot; u))\geq \gamma$ whenever $g(i)\neq g(j)$ for some large $\gamma$. Here, $\mathrm{BD}$ represents the Bhattacharyya distance $\mathrm{BD}(\cdot\| \cdot) =C_{0.5}(\cdot\|\cdot)$ \cite{fukunaga2013introduction}. 
Similarly, for each fixed $\theta\in\mathcal S_\Theta$, we define the clustering function $g_{\theta}$ for the conditional subcollection $I_\theta$: (1) $\mathrm{KL}(\mathcal N_{i,t}(\cdot;u)\|\mathcal N_{j,t}(\cdot; u))\leq \kappa_{\theta}$ whenever $g_\theta(i) = g_\theta(j)$ for some small $\kappa_\theta$ and (2) $\mathrm{BD}(\mathcal N_{i,t}(\cdot; u)\|\mathcal N_{j,t}(\cdot; u))\geq \gamma_{\theta}$ whenever $g_\theta(i) \neq g_\theta(j)$ for some large $\gamma_\theta$.

Under the above clustering conditions, the gap $\Delta I^{(N)}_u$ between $\widehat I^{(N)}_u(\Theta;Y_t\mid h_t)$ and particle MI $I^{(N)}_u(\Theta;Y_t\mid h_t)$ is formally upper bounded as in Corollary~\ref{prop:tightness}. The proof can be referred to in Appendix~\ref{appendix: Derivation of difference between MI and proposed upper bound}.

\begin{cor}[Tightness of the Upper Bound]\label{prop:tightness}
Let $|G|$ be the number of global clusters, $|G_\theta|$ be the number of clusters
within $\mathcal I_\theta$, and $\alpha,\alpha_\theta\in[0,1]$ be the Chernoff parameters. Then, 
\begin{align}
&\Delta \tilde I^{(N)}_{u}
\le
\kappa + (|G|-1)e^{\big(-(1-|1-2\alpha|)\gamma\big)}
+\notag\\&\sum_{\theta\in\mathcal S_\Theta}\hat q_t^{(N)}(\theta)\Big[
\kappa_\theta + (|G_\theta|-1)e^{\big(-(1-|1-2\alpha_\theta|)\gamma_\theta\big)}
\Big] \label{eq:MI_bound_tightness}
\end{align}  
\end{cor}
Corollary~\ref{prop:tightness} shows that the gap $\Delta I^{(N)}_u$ is controlled by 
the within-cluster similarity (via the KL terms, $\kappa,\kappa_{\theta}$) and
the between-cluster separation (via the BD terms, $\gamma,\gamma_{\theta}$). With a small within-cluster similarity and a large between-cluster separation, our established upper bound $\widehat{I}_u^{(N)}$ can be close to the particle MI $I_u^{(N)}$.

However, the KL/Chernoff mixture bound can be loose when strong privacy preservation is enforced. 
In particular, to reduce the distinguishability of $\Theta$ from the released data, the policy tends to produce similar mean and covariance across particles, i.e., $\mu_{i,t}(u)\approx \mu_{j,t}(u)$ and
$\Sigma_{i,t}(u)\approx \Sigma_{j,t}(u)$ even when $i$ and $j$ belong to different clusters (or different $\theta$-groups).
Accordingly, the within-cluster similarity parameters $\kappa$ and $\kappa_\theta$ may remain small, but the between-cluster separation measures $\gamma$ and $\gamma_\theta$ decrease substantially.
Note that \eqref{eq:MI_bound_tightness} contains terms of the form
$(|G|-1)e^{-(1-|1-2\alpha|)\gamma}$ and
$(|G_\theta|-1)e^{-(1-|1-2\alpha_\theta|)\gamma_\theta}$, which can be large if $\gamma$ and $\gamma_\theta$ are small.
In the extreme case where $\gamma,\gamma_\theta \approx 0$, \eqref{eq:MI_bound_tightness}  reduces to $\Delta I_u^{(N)}
\lesssim
\kappa + (|G|-1)+
\sum_{\theta\in\mathcal S_\Theta}\hat q_t^{(N)}(\theta)\Big[\kappa_\theta + (|G_\theta|-1)\Big]$,
which can be large whenever multiple clusters exist.
As a result, the KL/Chernoff mixture bound may remain high even when the true particle MI is already small,
making it overly conservative in high-privacy regimes.

To mitigate the looseness of the KL/Chernoff mixture bound in the high-privacy regime, we exploit a universal maximum-entropy upper bound. Specifically, for the marginal GMM $p^{(N)}(y\mid h_t;u)$ in~\eqref{eq:global_gmm_sum},
let $\bar\mu_t=\sum_{i=1}^N \omega_{i,t}\mu_{i,t}(u)$ and $\Sigma_{\mathrm{glob},t}
=\sum_{i=1}^N \omega_{i,t}\Big(\Sigma_{i,t}(u)
+\big(\mu_{i,t}(u)-\bar\mu_t\big)\big(\mu_{i,t}(u)-\bar\mu_t\big)^\top\Big)$. Since the Gaussian distribution maximizes entropy for a given covariance~\cite{10.5555/1146355}, we have $H^{(N)}(Y_t\mid h_t;u)
\le \widehat H^{(N)}_{\mathrm{G}}(Y_t\mid h_t;u)$, where $H^{(N)}_{\mathrm{G}}(Y_t\mid h_t;u)
= \frac{1}{2}\log((2\pi e)^d \,|\Sigma_{\mathrm{glob},t}|)$. Combining it with the Chernoff-based lower bound on
$H^{(N)}(Y_t\mid \Theta,h_t;u)$ in Theorem~\ref{th: upp-bound}, we propose the following regime-adaptive upper bound on the particle MI. In Appendix~\ref{appendix:proof for regime bound}, we prove that when the induced mixture components substantially overlap (e.g., $\mu_{i,t}(u)\approx \mu_{j,t}(u)$ and $\Sigma_{i,t}(u)\approx \Sigma_{j,t}(u)$ for many $(i,j)$), the Gaussian term $\widehat H^{(N)}_{\mathrm{G}}(Y_t\mid h_t;u)$ typically becomes active bound (i.e., minimizes the expression), yielding a tighter estimate.
\begin{thm}[Regime-Adaptive MI Upper Bound]\label{th:regime_adaptive}
The per-step particle MI is upper bounded by:
\begin{align}\label{eq:regime_adaptive_MI_ub}
I^{(N)}_u(\Theta;Y_t\mid h_t)
&\le
\widehat I^{(N)}_{u, \mathrm{RA}}(\Theta;Y_t\mid h_t),
\end{align}
where $\widehat I^{(N)}_{u, \mathrm{RA}}(\Theta;Y_t\mid h_t) =
\min\Big\{
\widehat H^{(N)}_{\mathrm{KL}}(Y_t\mid h_t;u),
\widehat H^{(N)}_{\mathrm{G}}(Y_t\mid h_t;u)
\Big\}
-
\sum_{\theta\in\mathcal S_\Theta}\hat q_t^{(N)}(\theta)
\widehat H^{(N)}_{C_\alpha}\big(Y_t\mid \theta,h_t;u\big)$.
\end{thm}
For all $u\in\mathcal U$, the regime-adaptive bound is never looser than the KL/Chernoff mixture bound in
Theorem~\ref{th: upp-bound}. 

\begin{remark}[Other cost terms]
    Under the Gaussian data-sharing policy, the deviation term $\mathbb{E}[d(X_t,Y_t)]$ and the expected system impact $\mathbb{E}[r(\theta^*, X_t, X_{t+1}, Y_t)]$ may admit analytical forms. For example, when the deviation is defined by the Euclidean distance, its expectation admits a closed-form expression via properties of the non-central $\chi^2$ distribution. In general, when analytical forms are unavailable, the corresponding expectations can be computed numerically via Monte Carlo sampling.
\end{remark}

We show in Theorem~\ref{lem: gmm-closed} (see Appendix~\ref{appendix: Lem_gmm-closed} for proof) that we can construct a similar closed-form upper bound on particle MI for the GMM policy.  

\begin{thm}
\emph{(\textbf{Upper Bound on particle MI under GMM Policy})}\label{lem: gmm-closed}
If the data-sharing policy is a finite Gaussian mixture for each particle, the closed-form upper bound on particle MI in Theorem~\ref{th:regime_adaptive} still applies.  
\end{thm}

Consequently, assuming $H(Y_t\mid h_t;u)$ and $H(Y_t\mid \theta,h_t;u)$ are finite for all $\theta$, by the universal approximation property of finite GMM (e.g., \cite{nguyen2020approximation,norets2010approximation}), our proposed upper bound of particle-based MI extends to any continuous policy. See Appendix~\ref{appendix:Cor_any_policy} for proof.

\begin{cor}\emph{(\textbf{Universal Closed-Form Upper Bound})}\label{cor:any_policy}
For any continuous policy kernel density
$k_t(\cdot\mid \theta,x,\hat\beta_t^{(N)};u)$
and any $\varepsilon>0$, there exists a finite Gaussian mixture policy kernel
$\hat{\mathcal K}_{J(\delta)}$ such that the induced per-step particle-based MI satisfy
\begin{align}
\Big|
I^{(N)}(\Theta;Y_t\mid h_t;\mathcal K_t)
-
I^{(N)}(\Theta;Y_t\mid h_t;\hat{\mathcal K}_{J(\delta)})
\Big|
\le \varepsilon.\notag
\end{align}
Moreover, applying Theorem~\ref{th:regime_adaptive} to $\hat{\mathcal K}_{J(\delta)}$ yields the computable upper bound
\begin{align}
&I^{(N)}(\Theta;Y_t\mid h_t;\mathcal K_t)
\le\widehat I^{(N)}_{\hat{\mathcal K}_{J(\delta)}, \mathrm{RA}}(\Theta;Y_t\mid h_t) +\varepsilon.
\notag
\end{align}

\end{cor}

Therefore, with these tractable expressions, the original privacy-utility optimization can be reformulated as a particle belief-state MDP, whose optimal data-sharing parameter $u^*$ satisfies the particle Bellman optimality equation~\eqref{eq:particle_bellman}. 
We next learn the parameterized data-sharing policy via reinforcement learning (RL).

\subsection{Particle filter RL-based solution method}
\label{sec: Particle Filter Reinforcement Learning-Based Solution Method}

With the particle representation of the belief state, we employ an actor–critic algorithm to learn the optimal parameter $u^*$, which induces the optimal policy $\mathcal{K}^*_t(Y_t\mid\theta, x_t,\hat \beta^{(N)}_t) = \mathcal{K}_t(Y_t\mid\theta, x_t,\hat \beta^{(N)}_t;u^*)$. Our proposed algorithm comprises two main components: an MGF encoder and a DDPG learning algorithm.

\emph{(1) Moment-Generating Function (MGF) encoder}. We adopt an MGF encoder \cite{ma2020discriminative} to address two challenges arising from particle-based belief representations. First, a large number of particles are often required in high-dimensional state spaces to achieve high-precision belief estimation, which may make training unstable. Second, the resampling step in PF introduces particle permutations, further complicating convergence.  
The MGF encoder addresses these challenges by capturing higher-order moments and characterizing distributions under mild regularity conditions~\cite{bulmer2012principles}. Specifically, we approximate the MGF of the particle belief $\hat\beta^{(N)}_t$ to obtain a low-dimensional vector
$s_t=\left[\bar s_t,M_{\hat\beta^{(N)}_t}^{1:m}\right]$, 
where $\bar s_t = \left[\sum_{i=1}^N\omega_{i,t}\theta_i,\sum_{i=1}^N\omega_{i,t}x_{i,t}\right]$ is the first-order moment of $\hat\beta^{(N)}_t$. Higher-order information of $\hat\beta^{(N)}_t$ is captured by MGF features evaluated at learned evaluation points $\{v^i\}^m_{i=1}$, yielding $M_{\hat\beta^{(N)}_t}^{1:m} = \left[M_{\hat\beta^{(N)}_t}^1, M_{\hat\beta^{(N)}_t}^2, \ldots, M_{\hat\beta^{(N)}_t}^m\right]$, where $M_{\hat\beta^{(N)}_t}^i$ is the MGF feature evaluated at $v^i$. The MFG encoder $S_t=\mathrm{Enc}_\phi(\hat\beta^{(N)}_t)$ is represented as a fully connected layer with ReLU activation functions, parametrized by $\phi$.

\emph{(2) DDPG-based learning algorithm}. 
We employ a DDPG algorithm involving an actor network and a critic network \cite{lillicrap2015continuous}. The actor network $(\hat\mu_t,\hat\Sigma_t) = f_u(\theta,x_t,s_t)$ instantiates the Gaussian policy $\mathcal K_t\big(\cdot \mid \theta,x_t,\hat\beta_t^{(N)};u\big)$ by generating its mean $\hat\mu_t$ and covariance matrix $\hat \Sigma_t$ from the current state $x_t$, private parameter $\theta$, and MFG features $s_t$. The shared data $y_t$ is then sampled accordingly.  
The critic network, parameterized by $w$, estimates the action-value function
$Q(\theta^*,s_t,x_t,\hat\mu_t,\hat\Sigma_t\,|\,w)$.  Target networks $Q'$ and $f_{u'}$ (and the target encoder $\mathrm{Enc}_{\phi'}$) are updated by Polyak averaging with rate $\tau\in(0,1)$. 
The critic is trained by minimizing the temporal-difference (TD) loss, and the actor parameters $u$ are updated using the deterministic policy gradient. 
The MGF encoder is trained jointly with the critic by backpropagating the TD loss through 
$s_t=\mathrm{Enc}_\phi(\hat\beta^{(N)}_t)$, encouraging a value-predictive (Bellman-consistent) representation. This joint training yields denser and lower-variance learning signals than actor-only updates.

\section{Experiments}
\label{sec:Simulation Settings and Results}
In this section, we conduct numerical simulations to evaluate the proposed privacy-preserving data-sharing framework in a mixed-autonomy platoon control scenario~\cite{yu2025interaction,ZHOU2024104885}.

\subsection{Experimental setup}
\label{subsec: experiment setup}
We consider a three-vehicle mixed-autonomy platoon comprising a leading vehicle (indexed by $0$), a connected and automated vehicle (CAV, indexed by $1$) in the middle, and a human-driven vehicle (HDV,indexed by $2$) at the tail. The HDV acts as a data provider, sharing its real-time spacing and speed to the CAV via vehicle-to-vehicle communication. The CAV uses this information to enhance platoon performance. Although the HDV may benefit from improved platoon operation, it has privacy concerns that the shared data could be exploited by adversaries to infer its driving behavior parameters, potentially leading to degraded safety and economic risks such as increased insurance premiums~\cite{ZHOU2024104885}.

\textbf{Simulation setup}. Similar to \cite{ZHOU2024104885,wang2021leading}, the system dynamics of the
CAV and HDV are defined as:
\begin{align*}
    \dot{s}_{i}(t) & = v_{i-1}(t) - v_{i}(t) + \zeta_{s,i}(t), i = 1,2 \\ 
    \dot{v}_{i}(t) &= 
    \begin{cases} 
        u(t), & i = 1, \\ 
        \mathbf{F}_{\theta^*}\left(s_i(t), v_i(t), v_{i-1}(t)\right) + \zeta_{a,i}(t), &i = 2.
    \end{cases}
\end{align*}
where disturbances $\zeta_a(t)$ and $\zeta_s(t)$ model the discretization error and human modeling error. The control action $u(t)$ for the CAV is the acceleration determined via a distributed linear controller $
    u(t) = \sum_{i=1}^3\left(\mu_{i}\left(s_{i}-s^{\star}\right) + \eta_{i}\left(v_{i}-v^{\star}\right)\right)$, where $\left(s^{\star},v^{\star}\right)$ represents the  equilibrium state, and $\mu_{i}$ and $\eta_{i}$ are feedback gains based on string stability conditions. The acceleration rate of HDV is governed by the car-following model $\mathbf{F}_{\theta^*}$, parameterized by a sensitive parameter $\theta^*$, which is an intrinsic property of a human driver. Here, we adopt the Full Velocity Difference Model (FVD)~\cite{jiang2001full} in the form $
     \mathbf{F}_{\theta^*} (\cdot) = \theta^*_1\left(V_{\text{FVD}}\left(s_3(t)\right)-v_3(t)\right) + \theta^*_2(v_{2}(t)-v_{3}(t))$, where $V_{\text{FVD}}(s)$ is the spacing-dependent desired velocity (see \cite{jiang2001full} for details).  

The deviation term in \eqref{eq: original_optimization} is chosen as the weighted Euclidean distance between shared data $Y_t = (\tilde{v}_t,\tilde{s}_t)$ and true data $X_t = (v_t,s_t)$  with weights $\omega_1\geq \omega_2$, i.e., 
$d(X_t,Y_t) = \sqrt{\omega_1\left(v_t-\tilde{v}_t \right)^2 +\omega_2\left(s_t-\tilde{s}_t\right)^2}$. In our experiment, we set $\omega_1 = 3$ and  $\omega_2 = 1$. The system cost $r\left(\theta^*, X_{t:t+1}, Y_t\right)$ is defined as the HDV's actual fuel consumption rate $f_i$(mL/s) \cite{bowyer1985guide}:
\begin{equation*}
f_i = \begin{cases} 
    0.444 + 0.090 R_i v_i + \left[0.054 a_i^2 v_i\right]_{a_i>0}, & R_i > 0 \\ 
    0.444, & R_i \leq 0 
\end{cases}
\end{equation*}
where $R_i = 0.333 + 0.00108 v_i^2 + 1.200 a_i$, with $v_i$ as the vehicle’s velocity and $a_i$ as its acceleration.

The system dynamics is discretized with a time step of $\Delta t=0.2\,\mathrm{s}$. The simulation environment has the following parameters: equilibrium spacing and velocity $(s^{*},v^{*})=(20\,\mathrm{m},15\,\mathrm{m/s})$, spacing thresholds $(s_{st},s_{go})=(5\,\mathrm{m},35\,\mathrm{m})$, acceleration limits $(a_{\min},a_{\max})=(-5,5)\,\mathrm{m/s}^2$, maximum velocity $v_{\max}=30\,\mathrm{m/s}$, and an AV velocity standard deviation  $\text{std}(\zeta_{a,i})=1\,\mathrm{m/s}^2$. 
For the AV, the linear controller feedback gains are $(\mu_1,\eta_1)=(0.1,0)$ for the AV, $(\mu_2,\eta_2)=(-0.5,0.5)$ for the CAV, and $(\mu_3,\eta_3)=(-0.2,0.2)$ for the HDV.

\textbf{Training configurations}. Each episode has a horizon $L=50$ and an upper bound on expected total distortion $\hat D=200$, with Lagrange multiplier $\lambda=0.2$, privacy weight $\rho=80$, and learned location-vector dimension $D^v=10$. The actor and critic learning rates are $lr_a=10^{-3}$ and $lr_c=10^{-3}$, both networks have 2 layers, and the MGF network consists of 1 layer.

\textbf{Benchmarks}. The proposed data-sharing policy is evaluated against: (i) True Data, where the HDV shares its true data with the CAV, and (ii) Discrete Policy ~\cite{weng2025optimal,erdemir2020privacy,yu2025interaction}, where the HDV adopts an advantage actor-critic-based data-sharing policy that discretizes both the true and shared data into grids. For the proposed method and the Discrete Policy benchmark, the network architectures include a 1-layer MLP for the MGF encoder and 2-layer MLPs for the actor and critic networks.

\subsection{Performance comparison}
\label{subsec: Performance comparison}

In this section, we evaluate the trained policy in terms of privacy preservation, data usability, and system-level cost (i.e., fuel consumption and control smoothness). The sensitive parameter $\theta^*$ is chosen from $\theta^1=(0.4,0.5)$, $\theta^2=(0.7,0.8)$, $\theta^3=(1.0,1.1)$, and $\theta^4=(1.3,1.4)$. All experiments are conducted over a horizon of $50$ time steps. To accurately assess privacy leakage, we use 5,184 particles to track the belief state.

\noindent\textbf{Parameter privacy}. 
Fig.~\ref{fig:posterior_indicator_pairs_2x2groups} shows, for each true parameter $\theta^*$, the evolution of the attacker's posterior mass (left figures) together with the detection outcomes of the maximum-likelihood occupancy estimator (right figures). 

We can see that the posterior mass on $\theta^*$ remains low across all cases, indicating limited privacy leakage, consistent with the frequent misdetections observed in the estimator output.

\begin{figure*}[t]
\centering
\captionsetup[figure]{font=footnotesize, skip=2pt}

\begin{subfigure}[b]{0.48\textwidth}
    \centering
    \includegraphics[width=0.49\linewidth]{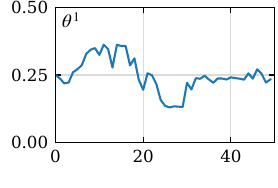}
    \includegraphics[width=0.49\linewidth]{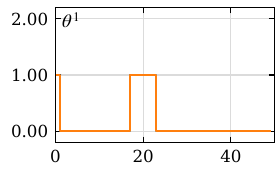}
    \caption{Results for $\theta^1$}
    \label{fig:group_theta1}
\end{subfigure}
\hfill 
\begin{subfigure}[b]{0.48\textwidth}
    \centering
    \includegraphics[width=0.49\linewidth]{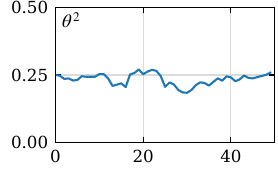}
    \includegraphics[width=0.49\linewidth]{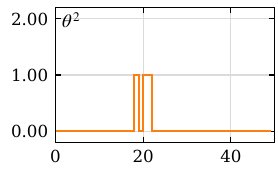}
    \caption{Results for $\theta^2$}
    \label{fig:group_theta2}
\end{subfigure}

\vspace{4mm} 

\begin{subfigure}[b]{0.48\textwidth}
    \centering
    \includegraphics[width=0.49\linewidth]{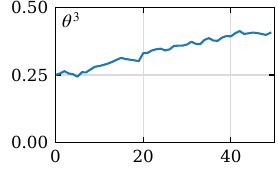}
    \includegraphics[width=0.49\linewidth]{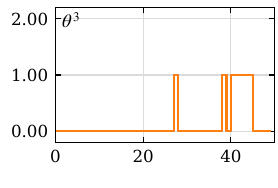}
    \caption{Results for $\theta^3$}
    \label{fig:group_theta3}
\end{subfigure}
\hfill
\begin{subfigure}[b]{0.48\textwidth}
    \centering
    \includegraphics[width=0.49\linewidth]{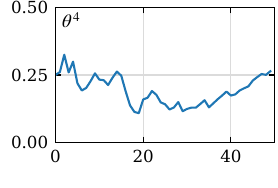}
    \includegraphics[width=0.49\linewidth]{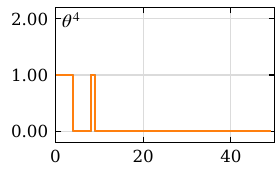}
    \caption{Results for $\theta^4$}
    \label{fig:group_theta4}
\end{subfigure}

\vspace{2mm}
\caption{Privacy performance. For each $\theta^\ast$, the left figure shows the evolution of posterior probability, while the right figure reports the misdetection instances over time.}
\label{fig:posterior_indicator_pairs_2x2groups}
\end{figure*}

\noindent\textbf{Data usability}. To assess data usability, we run a particle filter on the shared data to track the posterior of $x_t$ and report the posterior-mean estimate $\mathbb{E}[x_t\mid y_{1:t}]$. 
Fig.~\ref{fig:usability} (shown for $\theta^*=\theta^1$ due to space) indicates that the posterior mean closely matches the true velocity and spacing, with only a brief initial transient caused by prior initialization and limited early observations. 
We observe similar tracking performance for other values of $\theta^*$, confirming that the proposed data-sharing mechanism preserves state-estimation utility.
\begin{figure}[t]
    \centering

    \includegraphics[width=\linewidth]{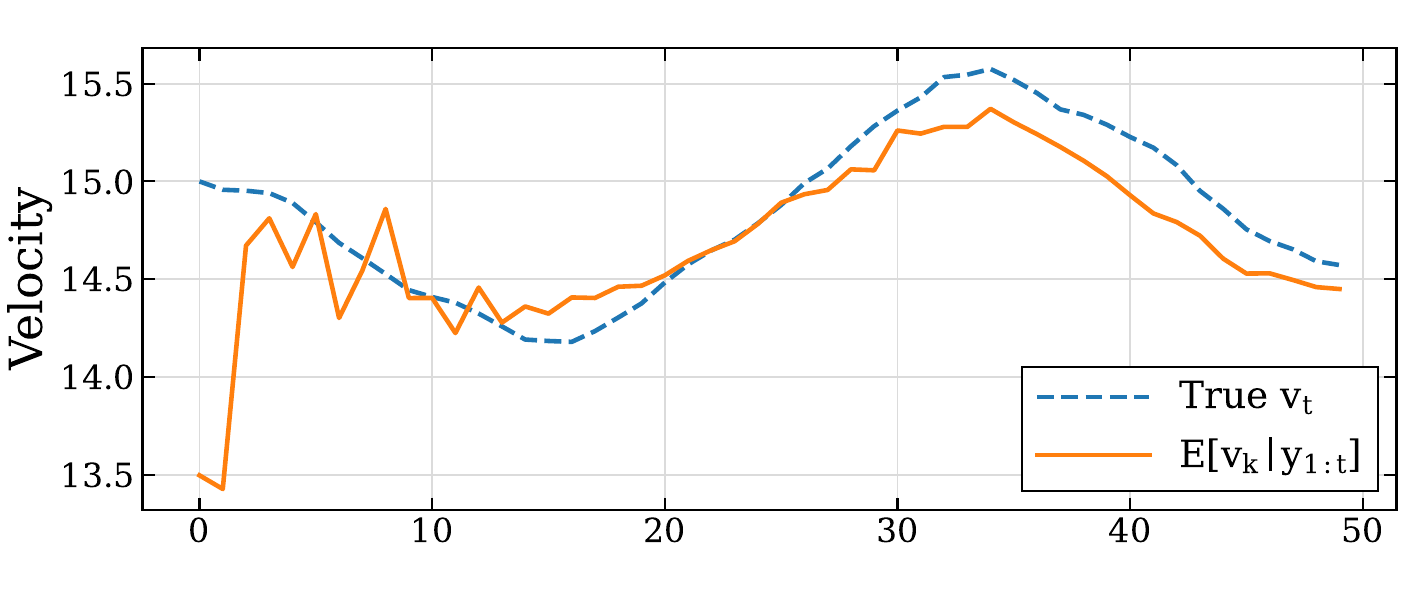}\\[-2pt]
    \vspace{2pt}

    \includegraphics[width=\linewidth]{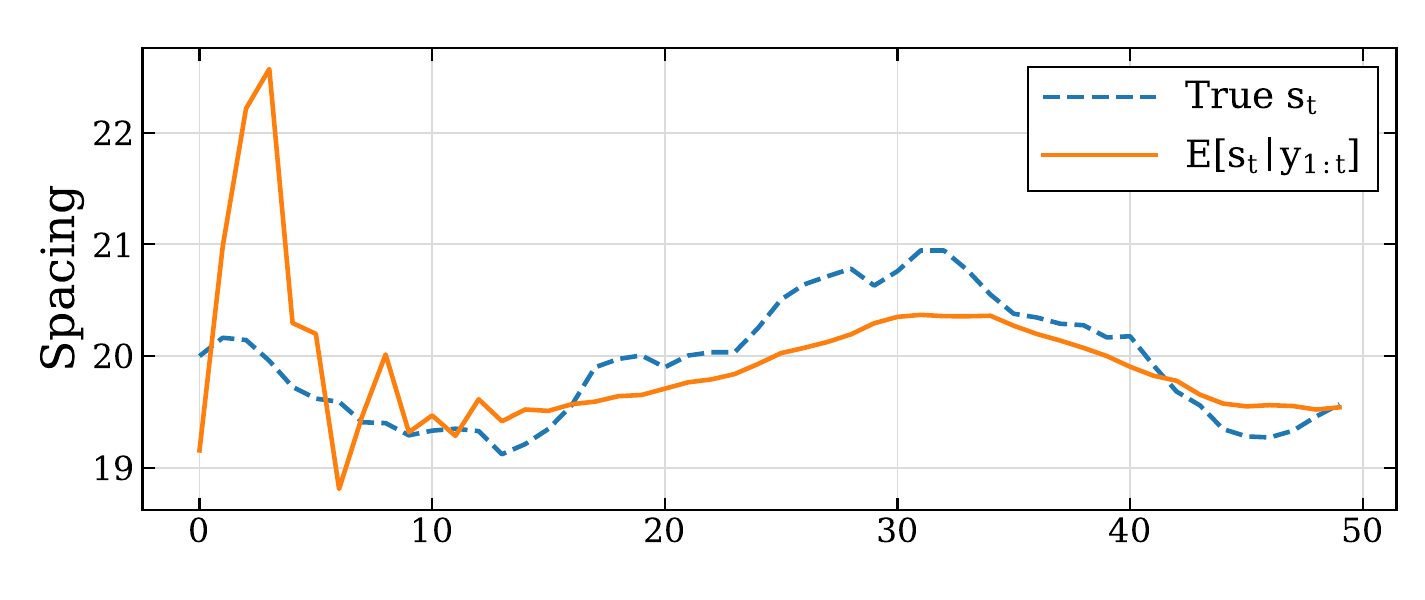}

    \caption{Data-usability: true states versus posterior mean estimates.}
    \label{fig:usability}
\end{figure}

\noindent\textbf{System cost}. 
We compare the HDV's fuel consumption performance (mL/s) and the smoothness of the resulting motion before and after applying the discrete and continuous policy. The fuel consumption results are summarized in Table~\ref{tab: fuel consumption}. The results indicate that for different driving behaviors $\Theta$, the continuous policy consistently yields lower expected fuel consumption than the discrete policy and stays very close to the “true data” baseline. Table~\ref{tab:acc-metrics} reports two smoothness measures: 
the average squared acceleration
$\mathbb{E}[a^2]$ and the average squared acceleration change $\mathbb{E}[(\Delta a)^2]$. We can see that across all true parameters, the continuous policy consistently reduces actuation intensity relative to the discrete policy. Specifically, $\mathbb{E}[a^2]$ under the continuous policy is about $1.15\times$--$1.55\times$ smaller than under the discrete policy, and is consistently closer to the true-data baseline.
Moreover, the continuous policy also yields smoother control, as $\mathbb{E}[(\Delta a)^2]$ is lower than that of the discrete policy for all $\theta^*$, again remaining closer to the true-data baseline. 
Overall, the continuous policy achieves lower acceleration magnitude and smaller acceleration variations than the discrete policy, indicating a clear system-cost advantage in terms of actuation efficiency and smoothness.
\begin{table}[htp]
  \centering
  \caption{Expected Fuel Consumption (ml/s)}
  \label{tab: fuel consumption}
  \small
  \setlength{\tabcolsep}{4pt}
  \begin{tabular}{lccc}
    \toprule

    $\theta^*$ & True Data & Discrete Policy & Continuous Policy \\ 
    \midrule
    $\theta^1$ & 1.24      & 1.29
          & 1.25 \\
    $\theta^2$ & 1.24
      & 1.30
          & 1.25 \\
    $\theta^3$ & 1.25
 & 1.34
 & 1.26 \\
    $\theta^4$ & 1.26
          & 1.35
              & 1.27     \\ 
    \bottomrule
  \end{tabular}
\end{table}

\begin{table}[htp]
  \centering
  \caption{Expected Acceleration Metrics}
  \label{tab:acc-metrics}
  
  \setlength{\tabcolsep}{3pt}                 

  \begin{tabular}{lcccccc}
    \toprule
    & \multicolumn{2}{c}{True Data}
    & \multicolumn{2}{c}{Discrete Policy}
    & \multicolumn{2}{c}{Continuous Policy} \\
    \cmidrule(lr){2-3}\cmidrule(lr){4-5}\cmidrule(lr){6-7}
    $\theta^*$ & $\mathbb{E}[a^2]$ & $\mathbb{E}[(\Delta a)^2]$
               & $\mathbb{E}[a^2]$ & $\mathbb{E}[(\Delta a)^2]$
               & $\mathbb{E}[a^2]$ & $\mathbb{E}[(\Delta a)^2]$ \\
    \midrule
    $\theta^1$ & 0.04 & 0.03 & 0.11 & 0.04 & 0.10 & 0.03 \\
    $\theta^2$ & 0.05 &0.04 & 0.13 & 0.06 & 0.11 & 0.05\\
    $\theta^3$ & 0.06 & 0.06 & 0.16 & 0.09 & 0.12 & 0.07 \\
    $\theta^4$ & 0.08     & 0.09    & 0.19     & 0.15     & 0.12    & 0.10\\
    \bottomrule
  \end{tabular}
\end{table}

\section{Conclusion}

In this paper, we studied parameter-privacy-preserving data sharing in continuous-state dynamical systems. We develop a particle-belief MDP for privacy-preserving data sharing by approximating the intractable belief dynamics via sequential Monte Carlo. The resulting formulation is tractable and asymptotically consistent as the number of particles grows. We further derive a regime-adaptive closed-form upper bound on particle-based MI using Gaussian-mixture approximations, enabling efficient optimization of continuous data-sharing policies. Experiments on a mixed-autonomy platoon demonstrated that the learned policy substantially impedes inference of human-driving behavior parameters while maintaining data usability and incurring negligible performance degradation. Future work will address privacy-aware high-dimensional data sharing and multi-stakeholder data-sharing mechanisms.

\begin{ack}                               This research was supported by the Singapore Ministry of Education (MOE) under its Academic Research Fund Tier 2 (A-8003064-00-00).
\end{ack}

\bibliographystyle{unsrt}        
\bibliography{autosam}

\appendix

\section{Proof for Simplification of MI in Section~\ref{sec:Problem Statement}}
\label{appendix: MI simplification}
By the chain rules of conditional MI, we have:
\begin{align}
    &I^{\boldsymbol{\pi}}(\Theta;Y_{1:T},W_{1:T})  
    =   \sum_{t=1}^T I^{\boldsymbol{\pi}}(\Theta;Y_t,W_{t}\mid Y_{1:t-1},W_{1:t-1}) \nonumber\\
     =~& \sum_{t=1}^T \Big(I^{\boldsymbol{\pi}}(\Theta; Y_t\mid Y_{1:t-1},W_{1:t}) \notag\\&\quad+ I^{\boldsymbol{\pi}}(\Theta; W_t\mid Y_{1:t-1},W_{1:t-1})\Big) \label{eq:MI1}
\end{align}
We only need to show $I^{\boldsymbol{\pi}}(\Theta; W_t \mid Y_{1:t-1}, W_{t-1}) = 0$. 
This is because $W_t$ is independent of $\Theta$ given $(W_{1:t-1}, Y_{1:t-1})$, as by the dynamics of $W_t$, the impact of $\Theta$ on $W_t$ is mediated through $(W_{1:t-1}, Y_{1:t-1})$.

\section{Proofs in  Sec.~\ref{sec: A Computationally Efficient Privacy-Preserving Policy}}

\subsection{Proof for Theorem~\ref{th:simplified policy Theorem}}
\label{appendix:simplified policy Theorem}
To prove Theorem~\ref{th:simplified policy Theorem}, we first show 
\begin{align}
p^{\pi}(\Theta,X_{t:t+1},W_{1:t},Y_{1:t})= p^{\pi_s}(\Theta,X_{t:t+1},W_{1:t},Y_{1:t}). \label{eq:full prob}
\end{align}
To this end, we decompose $ p^{\pi}(\Theta,X_{t:t+1},W_{1:t},Y_{1:t})$ as 
\begin{align}
& p^{\pi}(\Theta, X_{t:t+1}, W_{1:t}, Y_{1:t}) \notag\\
=~&p\!\left(X_{t+1} \mid \Theta, X_t, W_t, Y_{1:t}\right) \,
          p^{\pi}\!\left(\Theta, X_t, W_{1:t}, Y_{1:t}\right) \notag\\
 = ~&p\!\left(X_{t+1} \mid \Theta, X_t, W_t\right) \,
          \pi_t\!\left(Y_t \mid \Theta, X_t, W_{1:t}, Y_{1:t-1}\right) \,
          \notag\\&~ \cdot p^{\pi}\!\left(\Theta, X_t, W_{1:t}, Y_{1:t-1}\right) \notag\\
= ~&p\!\left(X_{t+1} \mid \Theta, X_t, W_t\right) \,
          \pi_t\!\left(Y_t \mid \Theta, X_t, W_{1:t}, Y_{1:t-1}\right) \notag\\
&~\cdot  p\!\left(W_t \mid W_{t-1}, Y_{t-1}\right) \,
          p^{\pi}\!\left(\Theta, X_t, W_{1:t-1}, Y_{1:t-1}\right) \notag\\
=~& p\!\left(X_{t+1} \mid \Theta, X_t, W_t\right) \,
          \pi_t\!\left(Y_t \mid \Theta, X_t, W_{1:t}, Y_{1:t-1}\right) \notag\\
&~\cdot  p\!\left(W_t \mid W_{t-1}, Y_{t-1}\right) \,
          \notag\\
&~\cdot  \int_{x_{t-1}} p^{\pi}\!\left(\Theta, X_{t-1}, x_t, W_{1:t-1}, Y_{1:t-1}\right) dx_{t-1}.
\label{eq:decomposion}
\end{align}
where the second equation follows the Markov property of the system dynamics of the data provider.

For any $\pi_t\in \Pi$, we can choose a simplified policy $\pi^s_t \in \Pi_s$  such that:
\begin{align}
    & \pi^s_t(Y_t|\Theta,X_t,W_{1:t},Y_{1:t-1})\notag\\=~& \pi_{t,Y_t|\Theta,X_t,W_{1:t},Y_{1:t-1}}(Y_t|\Theta,X_t,W_{1:t},Y_{1:t-1}).
\end{align}
We use induction to show \eqref{eq:full prob}. For $t=1$, $\pi_1(Y_1\mid \Theta, X_{1},W_{1}) = \pi^s_1(Y_1\mid \Theta, X_{1},W_{1})$ holds because we have $p^{\pi}(\Theta,X_{1},X_{2},W_{1},Y_{1}) = p(X_{2}\mid \Theta, X_{1},W_{1}) \pi_1(Y_1\mid \Theta, X_{1},W_{1})p(\Theta, X_{1},W_{1}) $. Thus, the induction can be initialized as $
    p^{\pi}(\Theta,X_{1},X_{2},W_{1},Y_{1}) = p^{\pi^s}(\Theta,X_{1},X_{2},W_{1},Y_{1})$. 
At the induction step, assuming $ p^{\pi}(\Theta,X_{t-1:t},W_{1:t-1},Y_{1:t-1})\allowbreak = p^{\pi_s}(\Theta,\allowbreak X_{t-1:t},W_{1:t-1},Y_{1:t-1})$, \eqref{eq:full prob} follows naturally from  \eqref{eq:decomposion}. 

With \eqref{eq:full prob}, we show the equivalence of MI. We have  
\begin{align}
    p^{\pi}(\Theta,W_{1:t},Y_{1:t}) &=\int_{x_{t:t+1}}p^{\pi}(\Theta,x_{t:t+1},W_{1:t},Y_{1:t}) dx_{t:t+1}\notag\\
    &= \int_{x_{t:t+1}}p^{\pi^s}(\Theta,x_{t:t+1},W_{1:t},Y_{1:t}) dx_{t:t+1}\notag\\&= p^{\pi^s}(\Theta,W_{1:t},Y_{1:t}), \label{eq:simplification_1}
\end{align}
which leads to $I^{\pi}(\theta; Y_{1:t}, W_{1:t}) = I^{\pi_s}(\theta; Y_{1:t}, W_{1:t})$ by the definition of MI. 

We then show the equivalence of control performance under policies $\pi$ and $\pi_s$, i.e., $\mathbb{E}^{\boldsymbol{\pi}}[r(\theta^*,X_t,X_{t+1},Y_t)]$. Note that 
\begin{align}
     &p^{\pi}(\theta^*,X_t,X_{t+1},Y_t)\notag \\ = &\int_{w_{1:t},y_{1:t-1}}p^{\pi}(\theta^*,X_{t:t+1},w_{1:t},Y_t,y_{1:t-1})dy_{1:t-1}dw_{1:t} \notag\\
     = & \int_{w_{1:t},y_{1:t-1}}p^{\pi_s}(\theta^*,X_{t:t+1},w_{1:t},Y_t,y_{1:t-1}) dy_{1:t-1}dw_{1:t}\notag\\
      = & p^{\pi^s}(\theta^*,X_t,X_{t+1},Y_t), 
\end{align}
where the second equation results from \eqref{eq:full prob}. Then, the equivalence of control performance can be proved by the definition of expectation. 

Similarly, we only need to prove $p^{\pi}(X_t,Y_t) = p^{\pi^s}(X_t,Y_t)$ for the equivalence of data deviation constraint under policies $\pi$ and $\pi_s$. This follows from \eqref{eq:simplification_1} by calculating the marginal probability on both sides of the equation.  

\subsection{Derivation of Belief State Update}
\label{appendix:update for belief}
The belief state update is derived using Bayes' rule. Starting with the posterior probability $\beta_{t+1}(\Theta,X_{t+1})$, we obtain the recursive update equation given in \eqref{eq: final_belief_update}.

\begin{figure*}[!t] 
    \centering
    \begin{align}
    & \beta_{t+1}(\Theta,X_{t+1}) = p(\Theta,X_{t+1}\mid W_{1:t+1},Y_{1:t}) = \frac{p(\Theta,X_{t+1}, W_{t+1}, Y_{t}\mid W_{1:t},Y_{1:t-1})}{p(W_{t+1},Y_t\mid  W_{1:t},Y_{1:t-1})}\notag\\
    &=\frac{p(W_{t+1} \mid W_{t},Y_{t})p(\Theta,X_{t+1}, Y_{t}\mid W_{1:t},Y_{1:t-1})}{p(W_{t+1} \mid W_{t},Y_{t})p(Y_t\mid  W_{1:t},Y_{1:t-1})} =\frac{p(\Theta,X_{t+1}, Y_{t}\mid W_{1:t},Y_{1:t-1})}{p(Y_t\mid  W_{1:t},Y_{1:t-1})}  \label{eq: forth_eq}\\
    & = \frac{\int_{x_t}p(\Theta,x_t,X_{t+1},Y_{t}\mid W_{1:t},Y_{1:t-1}) dx_t}{\int_{x_{t:t+1}}\int_{\theta}p(\theta,x_{t:t+1},Y_t\mid  W_{1:t},Y_{1:t-1})d\theta dx_{t:t+1}} = \frac{\int_{x_t}p(X_{t+1}\mid \Theta,x_t,W_t)p(\Theta,x_t,Y_{t}\mid W_{1:t},Y_{1:t-1})dx_t}{\int_{x_{t:t+1}}\int_{\theta}p(x_{t+1}\mid \theta,x_{t},W_t)p(\theta,x_{t},Y_t\mid  W_{1:t},Y_{1:t-1})d\theta dx_{t:t+1}} \notag\\
    &= \frac{\int_{x_t}p(X_{t+1}\mid \Theta,x_t,W_t) p(Y_{t}\mid \Theta,x_t,W_{1:t},Y_{1:t-1})p(\Theta,x_t\mid W_{1:t},Y_{1:t-1})dx_t}{\int_{x_{t:t+1}}\int_{\theta}p(x_{t+1}\mid \theta,x_{t},W_t) p(Y_{t}\mid \theta,x_t,W_{1:t},Y_{1:t-1})p(\theta,x_{t}\mid  W_{1:t},Y_{1:t-1}) d\theta dx_{t:t+1}}  \label{eq: belief_update} \\
    & = \frac{\int_{x_t}p(X_{t+1}\mid \Theta,x_t,W_t)\pi^s_t(Y_{t}\mid \Theta,x_t,W_{1:t},Y_{1:t-1})\beta_t(\Theta,x_t)dx_t}{\int_{x_{t:t+1}}\int_{\theta}p(x_{t+1}\mid \theta,x_{t+1},W_t)\pi^s_t(Y_{t}\mid \theta,x_t,W_{1:t},Y_{1:t-1})\beta_t(\theta,x_{t})d\theta dx_{t:t+1}} \label{eq: final_belief_update}.
\end{align}
\end{figure*}

Note that the third equation \eqref{eq: forth_eq} results from the following relationship 
\begin{align*}
    &p(W_{t+1} \mid \Theta, X_{t:t+1}, W_{1:t}, Y_{1:t})  = p(W_{t+1} \mid W_t, Y_t),
\end{align*}
which holds since $W_{t+1}$ is conditional independent of $(\Theta,X_{t:t+1},Y_{1:t-1},W_{1:t-1})$, given $(W_t,Y_t)$. This enables further decomposing $\beta_{t+1}(\Theta, X_{t+1})$ in \eqref{eq: final_belief_update}.

\subsection{Proof for Lemma~\ref{le:bellman equiv}}
\label{appendix:proof for bellman equiv}

According to policy $\pi^s_t$ and belief state $\beta_t$, the Bellman optimality equation \eqref{eq:Bellman optimality equation} can be decomposed into an immediate cost term and an expected future value. Since $h_t$ is given, we simplify the notation $\pi^s_t(Y_t\mid \Theta, X_t, h_t)$ to $\pi^s_t(Y_t\mid \Theta, X_t)$. The Bellman equation is given by:
\begin{align}\label{eq:Bellman optimality equation decomposed}
V^{*}_t(h_{t}) = \min_{\pi^{s}_t} \left\{ C_t(h_{t},\pi^s_t) + \mathbb{E}\left[V^{*}_{t+1}(h_{t+1}) \mid h_{t}, \pi^s_t \right] \right\},
\end{align}
where the expectation is taken with respect to the observation $Y_t$ and the exogenous noise $W_{t+1}$.
We define the immediate cost $C_t(h_{t},\pi^s_t)$ as the sum of the privacy loss (MI), the system cost, and the penalty of data usability:

\begin{align}
    &C_t(h_{t},\pi^t_s)  = \int_{y_t,\theta}\int_{x_t}\pi^{s}_t(Y_t\mid\theta,x_t)\beta_t(\theta,x_t)dx_t
\notag\\ &~ \cdot\log\frac{\int_{x_t}\pi^{s}_t(Y_t\mid\theta,x_t)\beta_t(\theta,x_t)dx_t}{\beta_t(\theta)\int_{\theta}\int_{x_t}\pi^{s}_t(Y_t\mid\theta,x_t)\beta_t(\theta,x_t)dx_td\theta} d\theta dy_t  \notag\\&~ +\int_{x_{t:t+1}}\int_{y_t}r(\theta^*,x_{t:t+1},y_t)p(x_{t+1}\mid \theta^*,x_t,w_t)
\notag\\ &~ \cdot \pi^{s}_t(Y_t\mid\theta^*,x_t)\beta_t(\theta^*,x_t)dy_tdx_{t:t+1} \notag\\&
~+ \lambda \int_{x_t}\int_{\theta,y_t}\pi^{s}_t(Y_t\mid\theta,x_t)\beta_t(\theta,x_t)d(x_t,y_t) dy_{t} d\theta dx_t \notag\\ & ~ - \lambda\hat{D}
\end{align}

From each part of $C_t$, we can see that given history $h_t$, $C_t$ is determined by belief state $\beta_t$ and policy $\pi^s_t$. Thus, we use $C_t(h_{t},\pi^t_s)$ to denote the immediate cost $C_t$ when taking policy set $\pi^s_t$ under history $h_t$.

Based on this dependency, we can use induction to assume that $V^{*}_{t+1}(h_{t+1})$ also depends on belief state $\beta_{t+1}$. Then, we rewrite the Bellman optimality equation as:
\begin{align}\label{eq:Bellman optimality equation}
    V^{*}_t(h_{t}) & = \min_{\pi^{s}_t}\{C_t(h_{t}, \pi^{s}_t)
   + \mathbb{E}(V^{*}_{t+1}(\beta_{t+1}\mid h_{t}))\}.
\end{align}
Using $\beta_{t+1} = \Phi(\beta_t,\pi^s_t,y_t)$ to compactly represent this recursion, we can decompose the second term $\mathbb{E}(V^{*}_{t+1}(\beta_{t+1}\mid h_{t}))$ as:
\begin{align}
    &\mathbb{E}(V^{*}_{t+1}(\beta_{t+1}\mid h_{t}))= \int_{y_t}p(y_t\mid h_{t})V^{*}_{t+1}(\beta_{t+1})dy_t
    \notag\\ =&\int_{\theta,x_t,y_t}\pi^s_t(y_t\mid\theta,x_t)\beta_t(\theta,x_t)V^{*}_{t+1}(\Phi(\beta_t,\pi^s_t,y_t))dx_td\theta dy_t,\notag
\end{align}
which also depends on $\beta_t$. Thus, combining this term with $C_t(h_{t},\pi^t_s)$ and using induction, we can prove that $ V^{*}_t(h_{t})$ depend on belief state $\beta_t$.

\section{Proofs of Section~\ref{sec: Particle-Belief Formulation}}

\subsection{Proof for Lemma~\ref{lem: weight_update}}\label{appendix:weight_update_derivation} 
We approximate the belief state $\beta_t$ by the particle measure $\hat{\beta}_t^{(N)}$ defined in~\eqref{eq:particle_measure_def},
i.e., $\hat{\beta}_t^{(N)}(d\theta,dx)=\sum_{i=1}^N \omega_{i,t}\delta_{(\theta_i,x_{i,t})}(d\theta,dx)$,
which provides a discrete weighted approximation to $p(\Theta,X_t\mid Y_{1:t},W_{1:t})$.

The weights $\{\omega_{i,t}\}_{i=1}^N$ are chosen using the principle of importance sampling. Suppose $p(\Theta,X_t\mid Y_{1:t-1},W_{1:t})$ is a probability density from which it is difficult to draw samples. We define the easier-to-sample proposal joint distribution (i.e. importance density) of $\theta_{i}$ and $X_t$ as  $q(\theta_{i},X_{i,t}\mid Y_{1:t-1},W_{1:t})$. We get the sample of the new distribution $(\theta_{i},X_{i,t})$ and update the weight $\omega_{i,t}$ as 
\( \displaystyle\omega_{i,t} \propto \frac{p(\theta_{i},X_{i,t}\mid Y_{1:t-1},W_{1:t})}{q(\theta_{i},X_{i,t}\mid Y_{1:t-1},W_{1:t})}\). The proposal joint distribution $q(\Theta,X_t\mid W_{1:t},Y_{1:t-1})$ is chosen to have the a factorized form following the Bayesian rule:
\begin{align}
    &q(\Theta,X_t\mid W_{1:t},Y_{1:t-1}) = q(X_t\mid \Theta,X_{t-1},W_{1:t},Y_{1:t-1})\notag\\&~\cdot q(\Theta,X_{t-1}\mid W_{1:t-1},Y_{1:t-2})
    \label{eq: q_fac}
\end{align}
Similarly, the joint density $p(\Theta,X_{t}\mid Y_{1:t},W_{1:t})$ can be factorized as:
\begin{align}
\label{eq: p_fac}
    &p(\Theta,X_{t}\mid Y_{1:t},W_{1:t})    
    \propto p(\Theta,X_{t-1}\mid W_{1:t-1},Y_{1:t-2})\notag\\ & ~\cdot p(W_{t}\mid W_{t-1},Y_{t-1})p(X_t\mid \Theta,X_{t-1},W_{t-1})\notag\\&~\cdot p(Y_{t-1}\mid \Theta,X_{t-1},W_{1:t-1},Y_{1:t-2})
\end{align}

Using the factorized form in Eqs.~\eqref{eq: q_fac} and \eqref{eq: p_fac}, we can derive the recursive weights update scheme as in \eqref{eq: weight update process}.
\begin{figure*}
\centering
\begin{align}
\label{eq: weight update process}
    \omega_{i,t} \propto \frac{\omega_{i,t-1}p(W_{t}\mid W_{t-1},Y_{t-1})p(Y_{t-1}\mid \theta_{i},X_{i,t-1},W_{1:t-1},Y_{1:t-2})p(X_{i,t}\mid \theta_{i},X_{i,t-1},W_{t-1})}{q(X_{i,t}\mid \theta_i,X_{i,t-1},W_{1:t},Y_{1:t-1})}
\end{align}
\end{figure*}

Assume the density $q(X_{i,t}\mid \theta_i,X_{i,t-1},W_{1:t}\allowbreak,Y_{1:t-1}) = p(X_{i,t}\mid \theta_i,X_{i,t-1}, W_{1:t-1})$, the final process of weights update can be concluded
\begin{align}
    \omega_{i,t} \propto& \omega_{i,t-1}p(W_{t}\mid W_{t-1},Y_{t-1})\notag\\&~\cdot p(Y_{t-1}\mid \theta_{i},X_{i,t-1},W_{1:t-1},Y_{1:t-2}).
\end{align} 
\subsection{Proof for Lemma~\ref{lem:MI-conv}}\label{appendix:lem_MI-conv}
We first express $C_t^{(N)}
(\hat\beta_t^{(N)},\mathcal K_t)$ explicitly in terms of the particle belief components. Since $\Theta$ takes values in a finite set, given a belief $\beta(\theta^j,dx)=q(\theta^j)b(dx\mid\theta^j)$, define the induced density
\begin{align}
&f_{Y\mid\Theta}(y\mid\theta^j;\beta,u) =\int k_t(y\mid\theta^j,x,\beta;u)b(dx\mid\theta^j),
\notag\\
&p_Y(y;\beta,u)=\sum_{j=1}^K q(\theta^j)f_{Y\mid\Theta}(y\mid\theta^j;\beta,u).\notag
\end{align}

Applying the particle belief $\hat\beta_t^{(N)}$, we obtain
\begin{align}
&f_{Y\mid\Theta}^{(N)}(y\mid\theta^j;u)
=\int k_t(y\mid\theta^j,x,\hat\beta_t^{(N)};u)\hat b_t^{(N)}(dx\mid\theta^j),\notag 
\\
&p_Y^{(N)}(y;u)
=\sum_{j=1}^K \hat q_t^{(N)}(\theta^j) f_{Y\mid\Theta}^{(N)}(y\mid\theta^j;u).\notag
\end{align}
Accordingly, the stepwise MI under the particle belief $\beta^{N}_t$ and  parameter $u$ is
\begin{align}\label{eq:MI-discrete-theta-simpler}
I^{(N)}(\Theta;Y_t\mid h_t,u)
= \sum_{j=1}^K &\hat q_t^{(N)}(\theta^j)
\int
f_{Y\mid\Theta}^{(N)}(y\mid\theta^j;u)\notag\\&
\log\frac{f_{Y\mid\Theta}^{(N)}(y\mid\theta^j;u)}{p_Y^{(N)}(y;u)}dy.
\end{align}

Similarly, the particle approximations of the system-cost term $C_{t,1}^{(N)}$ and the
data-distortion term $C_{t,2}^{(N)}$ are obtained by replacing $\beta_t$ with $\hat\beta_t^{(N)}$. Let $x'$ be a shorthand for $x_{t+1}$. Then
\begin{subequations}
    \begin{align}
        C_{t,1}^{(N)}(\hat\beta_t^{(N)},u)
=&\iiint r(\theta^*,x,x',y)
p(dx'\mid \theta^*,x,w_t)\notag\\&
\cdot k_t(y\mid \theta^*,x,\hat\beta_t^{(N)};u)dy
\hat\beta_t^{(N)}(\theta^*,dx) \notag
    \end{align}
    \begin{align}
        C_{t,2}^{(N)}(\hat\beta_t^{(N)},u)
=&\sum_{\theta\in\mathcal S_\Theta}\iint d(x,y)
k_t(y\mid \theta,x,\hat\beta_t^{(N)};u)\notag\\&\cdot dy
\hat\beta_t^{(N)}(\{\theta\},dx)
-\lambda\hat D \notag
    \end{align}
\end{subequations}

For brevity, we only show the convergence of $I^{(N)}(\Theta;Y_t\mid h_t,u)$. The proof of the other terms $ C_{t,1}^{(N)}(\hat\beta_t^{(N)},u)$ and $ C_{t,2}^{(N)}(\hat\beta_t^{(N)},u)$ are similar. 

Since $\mathcal S_\Theta=\{\theta^1,\ldots,\theta^K\}$ is finite, the joint weak convergence in Lemma~\ref{lem:pf-consistency} immediately yields the weak convergence of the marginal and
conditional belief components.

For each $j$, define the marginal measures on $\mathsf X$ by $\hat\beta_{t,j}^{(N)}(A):=\hat\beta_t^{(N)}(\{\theta^j\}\times A)
$ and $\beta_{t,j}(A):=\beta_t(\{\theta^j\}\times A)$.
Take any bounded continuous function $g$ on $\mathsf X$ and define
$h(\theta,x)=g(x)\mathbf 1_{\{\theta=\theta^j\}}$.
The indicator $\mathbf 1_{\{\theta=\theta^j\}}$ is continuous on $\Theta$,
so $h$ is also bounded continuous.
By $\hat\beta_t^{(N)}\Rightarrow\beta_t$ almost surely, we have 
\[\int h(\theta,x)\,\hat\beta_t^{(N)}(d\theta,dx)
\longrightarrow
\int h(\theta,x)\,\beta_t(d\theta,dx).\]
almost surely. Taking $g\equiv 1$ gives $\hat q_t^{(N)}(\theta^j)=\hat\beta_{t,j}^{(N)}(\mathsf X)\to
\beta_{t,j}(\mathsf X)=q_t(\theta^j)$ almost surely for all $j$.
Moreover, if $q_t(\theta^j)>0$
\begin{align}
     &\int g(x)\,\hat b_t^{(N)}(dx\mid\theta^j) =\frac{\int g\,d\hat\beta_{t,j}^{(N)}}{\hat q_t^{(N)}(\theta^j)} \notag\\&\longrightarrow \frac{\int g\,d\beta_{t,j}}{q_t(\theta^j)} =\int g(x)\,b_t(dx\mid\theta^j),\qquad \text{a.s.}\notag
\end{align}
which shows $\hat b_t^{(N)}(\cdot\mid\theta^j)\Rightarrow b_t(\cdot\mid\theta^j)$ almost surely.

Next, fix $\theta^j$ and $y$.
By Assumption~\ref{assum:regularity}(ii)--(iii), the map $x\mapsto k_t(y\mid\theta^j,x,\beta;u)$ is continuous on $\mathsf X$ and satisfies
$\underline k_t(y)\le k_t(y\mid\theta^j,x,\beta;u)\le \overline k_t(y)$ for all $x$.
Therefore, using the conditional weak convergence above, we obtain 
\begin{align}
    &f_{Y\mid\Theta}^{(N)}(y\mid\theta^j;u)
=\int k_t(y\mid\theta^j,x,\hat\beta_t^{(N)};u)\,\hat b_t^{(N)}(dx\mid\theta^j)
\notag\\&\longrightarrow
\int k_t(y\mid\theta^j,x,\beta_t;u)\,b_t(dx\mid\theta^j)
=f_{Y\mid\Theta}(y\mid\theta^j;u)\notag
\end{align}
almost surely. Combining this with $\hat q_t^{(N)}(\theta^j)\to q_t(\theta^j)$ yields the marginal
density convergence
\begin{align}
    &p_Y^{(N)}(y;u)=\sum_{j=1}^K \hat q_t^{(N)}(\theta^j)\,f_{Y\mid\Theta}^{(N)}(y\mid\theta^j;u)
\notag\\ \longrightarrow
&\sum_{j=1}^K q_t(\theta^j)\,f_{Y\mid\Theta}(y\mid\theta^j;u)
=p_Y(y;u),
\qquad \text{a.s.}\notag
\end{align}
Finally, the envelope bounds are inherited by averaging:
\begin{subequations}\label{eq: uniform_bound}
\begin{equation}
    \underline k_t(y)\le f_{Y\mid\Theta}^{(N)}(y\mid\theta^j;u)\le \overline k_t(y),
\end{equation}
\begin{equation}
    \underline k_t(y)\le p_Y^{(N)}(y;u)\le \overline k_t(y),
\end{equation}
\end{subequations}
and the same bounds hold for the limits.
Then define the pointwise log-ratio $\displaystyle \ell_N(\theta^j,y;u)
= \log\frac{f_{Y\mid\Theta}^{(N)}(y\mid\theta^j;u)}{p_Y^{(N)}(y;u)}$ and $\displaystyle\ell(\theta^j,y;u)
= \log\frac{f_{Y\mid\Theta}(y\mid\theta^j;u)}{p_Y(y;u)}$. For each $\theta^j$ and each $y$,
$\ell_N(\theta^j,y;u)\to \ell(\theta^j,y;u)$ almost surely.
Moreover, by the uniform bounds \eqref{eq: uniform_bound}, 
\begin{align}
|\ell_N(\theta^j,y;u)| &= \Big|\log\frac{f_{Y\mid\Theta;u}^{(N)}(y\mid\theta^j;u)}{p_Y^{(N)}(y;u)}\Big|\notag\\&\leq |\log \overline  k_t(y)| +|\log \underline k_t(y)|\notag
\end{align}
Therefore, the integrand in \eqref{eq:MI-discrete-theta-simpler} is dominated by $\displaystyle \Big|\hat q_t^{(N)}(\theta^j)\, f_{Y\mid\Theta}^{(N)}(y\mid\theta^j)\,\ell_N(\theta^j,y)\Big|
\le
\overline k_t(y)\Big(|\log \overline k_t(y)|+|\log \underline k_t(y)|\Big)$. Define \(g_t =\overline k_t(y)\Big(|\log \overline k_t(y)|+|\log \underline k_t(y)|\Big)\), and $g_t$ is integrable on $\mathbb{R}^N$ by Assumption~\ref{assum:regularity}(ii). Since $\ell_N(\theta^j,y)\to \ell(\theta^j,y)$ pointwise a.s., the dominated convergence theorem (DCT)
applied to \eqref{eq:MI-discrete-theta-simpler} yields
$I^{(N)}(\Theta;Y_t\mid h_t)\to I(\Theta;Y_t\mid h_t)$ almosy surely.

\subsection{Proof for Theorem~\ref{thm:asymptotic_convergence}} \label{app:proof_thm_asymp}

By the proof sketch, we only need to prove $\Delta_N^{(C)}$ and $\Delta_N^{(V)}$ uniformly converge to 0. By Lemma~\ref{lem:MI-conv}, $C^{(N)}_t(\hat\beta_t^{(N)},u)\to C_t(\beta_t,u)$ almost surely for each $u\in\mathcal U_t$.
The following lemma allows us to strengthen this pointwise convergence to uniform convergence over $\mathcal U_t$.

\begin{lem}[Pointwise-to-uniform upgrade]\label{lem:ptwise_to_unif}
Let $\mathcal U$ be compact. Suppose $f_N:\mathcal U\to\mathbb R$ are equicontinuous and
$f_N(u)\to f(u)$ pointwise on $\mathcal U$, where $f$ is continuous. Then
$\sup_{u\in\mathcal U}|f_N(u)-f(u)|\to 0$.
\end{lem}

\begin{pf}
Suppose the convergence is not uniform. Then there exist $\varepsilon>0$ and a subsequence
$\{f_{N_k}\}$ such that
\[
\sup_{u\in\mathcal U}|f_{N_k}(u)-f(u)|\ge \varepsilon,\qquad \forall k.
\]
Since $\mathcal U$ is compact and $\{f_N\}$ is equicontinuous, the Arzel\`a--Ascoli theorem yields a further subsequence
$\{f_{N_{k_j}}\}$ that converges uniformly to some $g\in C(\mathcal U)$.
In particular, $f_{N_{k_j}}(u)\to g(u)$ for every $u\in\mathcal U$.
But $f_{N_{k_j}}(u)\to f(u)$ pointwise by assumption, hence $g\equiv f$.
Therefore $\sup_{u\in\mathcal U}|f_{N_{k_j}}(u)-f(u)|\to 0$, which contradicts the choice of $\varepsilon$.
\end{pf}
Note that Assumption~\ref{assum:regularity}(iv) implies the Lipschitz continuity in $u$:
there exists $L_C<\infty$ such that for all $N$ and all $u,v\in\mathcal U_t$,
\begin{subequations}\label{eq:C-Lip}
\begin{align}
\big|C^{(N)}_t(\hat\beta_t^{(N)},u)-C^{(N)}_t(\hat\beta_t^{(N)},v)\big|
&\le L_C\|u-v\|,\\
\big|C_t(\beta_t,u)-C_t(\beta_t,v)\big|
&\le L_C\|u-v\|.
\end{align}
\end{subequations}
Since $\mathcal U_t$ is compact, the family
$\{u\mapsto C^{(N)}_t(\hat\beta_t^{(N)},u)\}_N$ is equicontinuous on $\mathcal U_t$.
Moreover, by Lemma~\ref{lem:MI-conv}, $C^{(N)}_t(\hat\beta_t^{(N)},u)\to C_t(\beta_t,u)$
for each fixed $u\in\mathcal U_t$, and $u\mapsto C_t(\beta_t,u)$ is continuous.
Therefore, Lemma~\ref{lem:ptwise_to_unif} yields
\[
\Delta_N^{(C)}= \sup_{u\in\mathcal U_t}\big|C^{(N)}_t(\hat\beta_t^{(N)},u)-C_t(\beta_t,u)\big|\to 0 \qquad \text{a.s.}.
\]

We next show $\Delta_N^{(V)}\rightarrow 0$. By the triangle inequality
\begin{align}
\Delta_N^{(V)}
&\le
\sup_{u_t\in\mathcal U_t}\Big|
\mathbb E\!\big[V_{t+1}^{*,(N)}(\hat\beta_{t+1}^{(N)})-V_{t+1}^{*}(\hat\beta_{t+1}^{(N)})\mid h_t,u_t\big]
\Big| \notag\\
&\quad+
\sup_{u_t\in\mathcal U_t}\Big|
\mathbb E\!\big[V_{t+1}^{*}(\hat\beta_{t+1}^{(N)})-V_{t+1}^{*}(\beta_{t+1})\mid h_t,u_t\big]
\Big| \notag \\
&\le \sup_{u_t\in\mathcal U_t}\Big(\mathbb E\!\Big[\big|V_{t+1}^{*,(N)}(\hat\beta_{t+1}^{(N)})-V_{t+1}^{*}(\beta_{t+1})\big|\mid h_t,u_t\Big]
 \notag\\ 
 & \quad+\mathbb E\!\Big[\big|V_{t+1}^{*}(\beta_{t+1})-V_{t+1}^{*}(\hat\beta_{t+1}^{(N)})\big|\mid h_t,u_t\Big]\Big)\notag  \\
& \quad+
\sup_{u_t\in\mathcal U_t}\Big|
\mathbb E\!\big[V_{t+1}^{*}(\hat\beta_{t+1}^{(N)})-V_{t+1}^{*}(\beta_{t+1})\mid h_t,u_t\big]\Big|. \notag
\end{align}
For the first term, we have (i) $V_{t+1}^{*,(N)}(\hat\beta_{t+1}^{(N)})\to V_{t+1}^{*}(\beta_{t+1})$ almost surely by the induction hypothesis (applied for this fixed $u_t$) and (ii) $V_{t+1}^{*}(\hat\beta_{t+1}^{(N)})\to V_{t+1}^{*}(\beta_{t+1})$ almost surely, due to the consistency of the particle belief update. Moreover, the finite-horizon value functions are uniformly bounded. Hence, by DCT, the conditional expectation in the first term converges to $0$ for each fixed $u_t$. To upgrade this pointwise to uniform convergence over $u_t\in\mathcal U_t$,
note that Assumption~\ref{assum:regularity}(iv) implies that the mapping
$u_t \mapsto \mathbb E\!\big[V_{t+1}^{*,(N)}(\hat\beta_{t+1}^{(N)})-V_{t+1}^{*}(\hat\beta_{t+1}^{(N)})
\mid h_t,u_t\big]$
is Lipschitz (hence equicontinuous) on $\mathcal U_t$ with a constant independent of $N$. Since $\mathcal U_t$ is compact,  Lemma~\ref{lem:ptwise_to_unif} yields that the first term converges to $0$
uniformly over $u_t\in\mathcal U_t$. 

For the second term, the same argument applies. For each fixed $u_t$, the consistency of the
particle belief update implies $V_{t+1}^{*}(\hat\beta_{t+1}^{(N)})\to V_{t+1}^{*}(\beta_{t+1})$ a.s.,
and the DCT yields pointwise convergence of the corresponding conditional expectation to $0$.
Moreover, Assumption~\ref{assum:regularity}(iv) implies a uniform (in $N$) Lipschitz property in $u_t$,
so by Lemma~\ref{lem:ptwise_to_unif} the convergence is uniform over $u_t\in\mathcal U_t$.

Combining the above yields
$\Delta_N\to 0$ almost surely. Hence by \eqref{eq:delta_N}, we obtain $| V_t^{*,(N)}(\hat\beta_t^{(N)})-V_t^*(\beta_t)|\to 0 \qquad$ almost surely. This completes the induction step and the proof.
\qed

\subsection{Proof for Corollary~\ref{prop:tightness}}\label{appendix: Derivation of difference between MI and proposed upper bound}
The gap between the particle-based MI and its tractable upper bound is
\begin{align}
&\tilde I^{(N)}(u)
= \widehat H^{(N)}_{\mathrm{KL}}(Y_t\mid h_t;u)
\notag\\&\quad-\sum_{\theta\in\mathcal S_\Theta}\hat q_t^{(N)}(\theta)
\widehat H^{(N)}_{C_\alpha}\big(Y_t\mid \theta,h_t;u\big)
 - I^{(N)}(\Theta;Y_t\mid h_t;u)\notag\\
&= \Big(\widehat H^{(N)}_{\mathrm{KL}}(Y_t\mid h_t;u)-H^{(N)}(Y_t\mid h_t;u)\Big)\notag\\
&\quad-\sum_{\theta\in\mathcal S_\Theta}\hat q_t^{(N)}(\theta)\Big(\widehat H^{(N)}_{C_\alpha}(Y_t\mid \theta,h_t;u)-H^{(N)}(Y_t\mid \theta,h_t;u)\Big)\notag\\
&\le \Big(\widehat H^{(N)}_{\mathrm{KL}}(Y_t\mid h_t;u)- H^{(N)}_{C_\alpha}(Y_t\mid h_t;u)\Big)\notag\\
&\quad-\sum_{\theta\in\mathcal S_\Theta}\hat q_t^{(N)}(\theta)\Big(\widehat H^{(N)}_{C_\alpha}(Y_t\mid \theta,h_t;u)- H^{(N)}_{\mathrm{KL}}(Y_t\mid \theta,h_t;u)\Big).\notag
\end{align}
Under the cluster conditions, \cite{kolchinsky2017estimating} have already proved that    $\widehat H^{(N)}_{\mathrm{KL}}(Y_t\mid h_t;u)- H^{(N)}_{C_\alpha}(Y_t\mid h_t;u) \leq \kappa + (|G|-1)e^{\big(-(1-|1-2\alpha|)\gamma\big)}$.
When $\theta$ is conditioned, we have $\widehat H^{(N)}_{C_\alpha}(Y_t\mid \theta,h_t;u)- H^{(N)}_{\mathrm{KL}}(Y_t\mid \theta,h_t;u)\leq \kappa_\theta + (|G_\theta|-1)e^{\big(-(1-|1-2\alpha_\theta|)\gamma_\theta\big)}$. These yield
\begin{align}
&\tilde I^{(N)}(u)
\le
\kappa + (|G|-1)e^{\big(-(1-|1-2\alpha|)\gamma\big)}
+\notag\\&\sum_{\theta\in\mathcal S_\Theta}\hat q_t^{(N)}(\theta)\Big[
\kappa_\theta + (|G_\theta|-1)e^{\big(-(1-|1-2\alpha_\theta|)\gamma_\theta\big)}
\Big].\notag
\end{align}\qed
\subsection{Proof for Theorem~\ref{th:regime_adaptive}}
\label{appendix:proof for regime bound}
Consider a homoscedastic Gaussian mixture with common covariance $\Sigma$, weights $\omega_i$, and means $\mu_i$. Let $S = \sum_{i=1}^N \omega_i(\mu_i - \bar{\mu})(\mu_i - \bar{\mu})^\top$ be the between-class scatter matrix. We define the separation matrix $A = \Sigma^{-1}S$ and consider the regime where $\|A\|_F \to 0$ (implying $\mu_i \approx \mu_j$). The total covariance of the mixture is $\Sigma_{\mathrm{glob}} = \Sigma + S$. Then $\widehat H_G$ can be calculated by $\widehat H_{\mathrm{G}}= \frac{1}{2}\log\det(2\pi e (\Sigma + S)) = \frac{1}{2}\log\det(2\pi e \Sigma) + \frac{1}{2}\log\det(I + \Sigma^{-1}S)= H_0 + \frac{1}{2}\log\det(I + A)$, where $H_0$ is the entropy of a single component. Using the Taylor expansion $\log\det(I+A) = \mathrm{tr}(A) - \frac{1}{2}\mathrm{tr}(A^2) + O(\|A\|^3)$, we obtain:
\begin{equation}\label{eq:HG_expansion}\widehat H_{\mathrm{G}} = H_0 + \frac{1}{2}\mathrm{tr}(A) + O(|A|^2).
\end{equation}
The KL-based bound is given by $\widehat H_{\mathrm{KL}} = H_0 - \sum_{i=1}^N \omega_i \log \left( \sum_{j=1}^N \omega_j e^{-D_{ij}} \right)$, where $D_{ij} = \mathrm{KL}(\mathcal{N}_i \| \mathcal{N}_j) = \frac{1}{2}(\mu_i - \mu_j)^\top \Sigma^{-1} (\mu_i - \mu_j)$. Note that $D_{ij} = O(\|A\|)$. Using the expansion $e^{-x} \approx 1 - x$, the inner sum becomes $\sum_{j} \omega_j e^{-D_{ij}} = 1 - \sum_{j} \omega_j D_{ij} + O(\|A\|^2)$. Applying $-\log(1-x) \approx x$, $\widehat H_{\mathrm{KL}}$ can be approximated by $\widehat H_{\mathrm{KL}}\approx H_0 + \sum_{i} \omega_i \left( \sum_{j} \omega_j D_{ij} \right) = H_0 + \sum_{i,j} \omega_i \omega_j \frac{1}{2}(\mu_i - \mu_j)^\top \Sigma^{-1} (\mu_i - \mu_j)$.
Using $\sum_{i,j} \omega_i \omega_j (\mu_i - \mu_j)(\mu_i - \mu_j)^\top = 2S$ in the trace:
\begin{align}
\label{eq:HKL_expansion}
\widehat H_{\mathrm{KL}} &= H_0 + \frac{1}{2}\mathrm{tr}(\Sigma^{-1}(2S)) + O(|A|^2)\notag\\& = H_0 + \mathrm{tr}(A) + O(|A|^2).
\end{align}
Comparing \eqref{eq:HG_expansion} and \eqref{eq:HKL_expansion}, the difference is: $\widehat H_{\mathrm{KL}} - \widehat H_{\mathrm{G}} = \left( H_0 + \mathrm{tr}(A) \right) - \left( H_0 + \frac{1}{2}\mathrm{tr}(A) \right) = \frac{1}{2}\mathrm{tr}(A) + O(\|A\|^2)$.  Since $\Sigma$ and $S$ are positive semi-definite, $\mathrm{tr}(A) = \mathrm{tr}(\Sigma^{-1}S) \ge 0$. Thus, for sufficiently small separation (where the first-order term dominates), we have $\widehat H_{\mathrm{G}} \le \widehat H_{\mathrm{KL}}$.
\qed

\subsection{Proof for Theorem~\ref{lem: gmm-closed}}\label{appendix: Lem_gmm-closed}
For particle $i$, suppose the data-sharing policy is a finite Gaussian mixture, i.e., $k_t\big(y\mid \theta_i,x_{i,t},\hat\beta_t^{(N)};u\big) =\sum_{k=1}^{K_i} a_{i,k,t}(u)
\mathcal N\big(y;\mu_{i,k,t}(u),\Sigma_{i,k,t}(u)\big)$, where nonnegative mixture weights satisfy $\sum_{k=1}^{K_i} a_{i,k,t}(u)=1$, and the Gaussian parameters 
$\big(\mu_{i,k,t}(u),\Sigma_{i,k,t}(u)\big)$ may depend on $u$. The history-conditioned marginal density of $Y_t$ is
\begin{align}
p^{(N)}(y\mid h_t;u)
&=\iint k_t(y\mid \theta,x,\hat\beta_t^{(N)};u)\,\hat\beta_t^{(N)}(d\theta,dx)\notag\\
&=\sum_{i=1}^N \omega_{i,t}\,k_t\!\big(y\mid \theta_i,x_{i,t},\hat\beta_t^{(N)};u\big)\notag\\
&=\sum_{i=1}^N\sum_{k=1}^{K_i}\tilde\omega_{i,k,t}(u)\,\mathcal N_{i,k,t}(y;u),\notag
\end{align}
where $\mathcal N_{i,k,t}(\cdot;u) = \mathcal N\!\big(\cdot;\mu_{i,k,t}(u),\Sigma_{i,k,t}(u)\big)$ is the component density of particle $i$ and $\tilde\omega_{i,k,t}(u) = \omega_{i,t}\,a_{i,k,t}(u)$ indicates its (global) mixture weights, which satisfies $\sum_{i=1}^N\sum_{k=1}^{K_i}\tilde\omega_{i,k,t}(u)=1$. 
Hence $p^{(N)}(y\mid h_t;u)$ is itself a single finite GMM whose Gaussian components are
$\{\mathcal N_{i,k,t}(\cdot;u)\}_{i,k}$ with weights $\{\tilde\omega_{i,k,t}(u)\}_{i,k}$.

Similarly, the conditional density of $Y_t$ given $\Theta=\theta$ is
\begin{align}
&p^{(N)}(y\mid \theta,h_t;u)
=\int k_t(y\mid \theta,x,\hat\beta_t^{(N)};u)\,\hat b_t^{(N)}(dx\mid \theta)\notag\\
&=\sum_{i\in\mathcal I_\theta}\bar\omega_{i,t}(\theta)\,
k_t\!\big(y\mid \theta_i,x_{i,t},\hat\beta_t^{(N)};u\big)\notag\\
&=\sum_{i\in\mathcal I_\theta}\sum_{k=1}^{K_i}\tilde\omega_{i,k,t}(\theta;u)\,\mathcal N_{i,k,t}(y;u),\notag
\end{align}
where $\tilde\omega_{i,k,t}(\theta;u)
=\bar\omega_{i,t}(\theta)a_{i,k,t}(u)$ is the conditional mixture weight of particle $i$, satisfying $\sum_{i\in\mathcal I_\theta}\sum_{k=1}^{K_i}\tilde\omega_{i,k,t}(\theta;u)=1$. Thus, $p^{(N)}(y\mid \theta,h_t;u)$ is also a finite GMM (restricted to indices
$i\in\mathcal I_\theta$). Applying Theorem~\ref{th:regime_adaptive} to these mixtures yields an explicit, computable upper bound for $I^{(N)}(\Theta;Y_t\mid h_t)$. This proves Lemma~\ref{lem: gmm-closed}.
\qed

\subsection{Proof for Corollary~\ref{cor:any_policy}}\label{appendix:Cor_any_policy}
Let $\varepsilon>0$ be arbitrary.
By the $L^1$/total-variation denseness of finite GMM for conditional densities
\cite{nguyen2020approximation,norets2010approximation}, there exists a finite GMM policy kernel $\hat{\mathcal K}_{J(\delta)}$ with density
$\hat k_{t,J(\delta)}(\cdot\mid \theta,x,\hat\beta_t^{(N)})$ such that, for some $\delta>0$,
\begin{align}
\sup_{(\theta,x)}
\Big\|
k_t(\cdot\mid \theta,x,\hat\beta_t^{(N)})-\hat k_{t,J(\delta)}(\cdot\mid \theta,x,\hat\beta_t^{(N)})
\Big\|_{1}
\le \delta .\notag
\end{align}
Under Assumption~\ref{assum:regularity}, conditional entropy and MI are $L^1$/TV-continuous. Hence there exists a modulus of continuity $\omega(\cdot)$ with $\omega(\delta)\to0$
as $\delta\to0$ such that
\begin{align}
\Big|
I^{(N)}(\Theta;Y_t\mid h_t;\mathcal K_t)
-
I^{(N)}(\Theta;Y_t\mid h_t;\hat{\mathcal K}_{J(\delta)})
\Big|
\le \omega(\delta).
\label{eq:MI_continuity}
\end{align}
Applying Theorem~\ref{th:regime_adaptive} to the finite GMM policy $\hat{\mathcal K}_{J(\delta)}$ yields the computable upper bound
\begin{align}
&I^{(N)}(\Theta;Y_t\mid h_t;\hat{\mathcal K}_{J(\delta)})
\le \widehat I^{(N)}_{\hat{\mathcal K}_{J(\delta)}, \mathrm{RA}}(\Theta;Y_t\mid h_t).
\label{eq:MI_bound_GMM}
\end{align}
Combining \eqref{eq:MI_continuity} and \eqref{eq:MI_bound_GMM} gives
\begin{align}
&I^{(N)}(\Theta;Y_t\mid h_t;\mathcal K_t)
\le 
\widehat I^{(N)}_{\hat{\mathcal K}_{J(\delta)}, \mathrm{RA}}(\Theta;Y_t\mid h_t)
+\omega(\delta).\notag
\end{align}
Finally, choose $\delta$ such that $\omega(\delta)\le\varepsilon$.
This completes the proof.\qed

\end{document}